\documentclass[journal]{IEEEtran}
\usepackage[utf8]{inputenc}
\usepackage{color}
\usepackage{xcolor}
\usepackage{array}
\usepackage{verbatim}
\usepackage{float}
\usepackage{amsmath}
\usepackage{amsthm}
\usepackage{amssymb}
\usepackage{graphicx}
\usepackage{longtable}
\usepackage{multirow}
\usepackage{booktabs}
\usepackage[unicode=true,
bookmarks=false,
breaklinks=false,pdfborder={0 0 1},colorlinks=false]
{hyperref}
\hypersetup{
	colorlinks,bookmarksopen,bookmarksnumbered,citecolor=blue,urlcolor=blue}
\usepackage{cite}

\usepackage{lipsum}
\usepackage{mathtools}
\usepackage{cuted}
\providecommand{\tabularnewline}{\\}
\usepackage{algorithmic}
\usepackage{longtable}

\floatstyle{ruled}
\newfloat{algorithm}{tbp}{loa}
\providecommand{\algorithmname}{Algorithm}
\floatname{algorithm}{\protect\algorithmname}

\makeatletter
\let\oldforeign@language\foreign@language
\DeclareRobustCommand{\foreign@language}[1]{%
	\lowercase{\oldforeign@language{#1}}}

\let\oldforeign@language\foreign@language
\DeclareRobustCommand{\foreign@language}[1]{%
	\lowercase{\oldforeign@language{#1}}}

\ifCLASSINFOpdf
\else
\fi

\newcommand{\MYfooter}{\smash{
		\hfil\parbox[t][\height][t]{\textwidth}{\centering
			\thepage}\hfil\hbox{}}}

\def\ps@IEEEtitlepagestyle{%
	\def\@oddhead{\parbox[t][\height][t]{\textwidth}{\centering \scriptsize
			Personal use of this material is permitted. Permission from the author(s) and/or copyright holder(s), must be obtained for all other uses. Please contact us and provide details if you believe this document breaches copyrights.\\
			\noindent\makebox[\linewidth]{}
		}\hfil\hbox{}}%
	\def\@evenhead{\scriptsize\thepage \hfil \leftmark\mbox{}}%
	\def\@oddfoot{\parbox[t][\height][l]{\textwidth}{
			\vspace{-20pt}{\rule{\textwidth}{0.4pt}}\\ \footnotesize\underline{To cite this article:}
			{\bf{\footnotesize\textcolor{red}{H. A. Hashim "Exponentially Stable Observer-based Controller for VTOL-UAVs without Velocity Measurements," International Journal of Control, vol. 96, no. 8, pp. 1946-1960, 2023.}}} doi: \href{https://doi.org/10.1080/00207179.2022.2079004}{10.1080/00207179.2022.2079004}\\
			\noindent\makebox[\linewidth]
		}\hfil\hbox{}}%
	\def\@evenfoot{\MYfooter}}

\makeatother
\newtheorem{defn}{Definition}

\newtheorem{lem}{Lemma}

\newtheorem{thm}{Theorem}
\newtheorem{rem}{Remark}

\newtheorem{assum}{Assumption}

\begin{document}
	\bstctlcite{IEEEexample:BSTcontrol}

	\title{Exponentially Stable Observer-based Controller for VTOL-UAVs without Velocity Measurements}

\author{Hashim A. Hashim% <-this % stops a space
	\thanks{This work was supported in part by National Sciences and Engineering
		Research Council of Canada (NSERC), under the grants RGPIN-2022-04937 and DGECR-2022-00103.}
	\thanks{H. A. Hashim is with the Department of Mechanical and Aerospace Engineering,
		Carleton University, Ottawa, Ontario, K1S-5B6, Canada, email: hhashim@carleton.ca}
}

% \markboth{IEEE TRANSACTIONS ON INTELLIGENT TRANSPORTATION SYSTEMS, \today}{Hashim \MakeLowercase{\textit{et al.}}: Landmark and IMU Data Fusion: Systematic Convergence Geometric Nonlinear Observer for SLAM and Velocity Bias}

% \markboth{}{Hashim \MakeLowercase{\textit{et al.}}: Nonlinear Filter for Simultaneous Localization and Mapping on a Matrix Lie Group using IMU and Feature Measurements}

\maketitle

\begin{abstract}
There is a great demand for vision-based robotics solutions that can
operate using Global Positioning Systems (GPS), but are also robust
against GPS signal loss and gyroscope failure. This paper investigates
the estimation and tracking control in application to a Vertical Take-Off
and Landing (VTOL) Unmanned Aerial Vehicle (UAV) in six degrees of
freedom (6 DoF). A full state observer for the estimation of VTOL-UAV
motion parameters (attitude, angular velocity, position, and linear
velocity) is proposed on the Lie Group of $\mathbb{SE}_{2}\left(3\right)\times\mathbb{R}^{3}$ $=\mathbb{SO}\left(3\right)\times\mathbb{R}^{9}$ with almost globally exponentially stable closed loop error signals.
Thereafter, a full state observer-based controller for the VTOL-UAV
motion parameters is proposed on the Lie Group with a guaranteed almost
global exponential stability. The proposed approach produces good
results without the need for angular and linear velocity measurements
(without a gyroscope and GPS signals) utilizing only a set of known
landmarks obtained by a vision-aided unit (monocular or stereo camera).
The equivalent quaternion representation on $\mathbb{S}^{3}\times\mathbb{R}^{9}$ is provided in the Appendix. The observer-based controller is presented
in a continuous form while its discrete version is tested using a
VTOL-UAV simulation that incorporates large initial error and uncertain
measurements. The proposed observer is additionally tested experimentally
on a real-world UAV flight dataset.
\end{abstract}


\IEEEpeerreviewmaketitle{}

\section{Introduction}

\IEEEPARstart{S}{uccessful} inertial navigation of Unmanned Aerial Vehicles (UAVs),
underwater vehicles, and ground vehicles, among other engineering
applications, requires robust solutions for attitude (orientation),
angular velocity, position, and linear velocity estimation and tracking
control. It has long been recognized that rigid-body motion parameters
cannot be obtained directly, but are instead reconstructed from sensor
measurements. Rigid-body's attitude can be reconstructed algebraically
using inertial-frame observations and corresponding body-frame measurements
\cite{markley1988attitude,mortari2000second} followed by position
reconstruction. However, better attitude and pose estimation solutions
are offered by Kalman filters \cite{markley2003attitude,chang2017indirect,pham2015gain},
nonlinear filters on the \textit{Special Orthogonal Group} $\mathbb{SO}\left(3\right)$
that mimic the true attitude dynamics geometry \cite{zlotnik2016exponential,hashim2019SO3Wiley},
and nonlinear filters on the \textit{Special Euclidean Group} $\mathbb{SE}\left(3\right)$
that mimic the true pose (attitude + position) dynamics geometry \cite{hashim2020SE3Stochastic,moeini2016global}.
The true attitude dynamics representation relies on angular velocity
typically measured by a gyroscope. The true pose, on the other hand,
requires measurements of both angular and linear velocity where the
linear velocity is generally made available by the Global Positioning
System (GPS) sensors \cite{hashim2020SE3Stochastic,hashim2021Navigation,moeini2016global,pesce2020radial,lozano2021pvtol,xie2016state,hashim2020LetterSLAM}.
Hence, obtaining linear velocity in a GPS-denied region poses a challenge
\cite{hashim2021Navigation,scaramuzza2014vision,qin2019autonomous,hashim2021gps}.
Thereby, the filters in \cite{markley2003attitude,chang2017indirect,zlotnik2016exponential,hashim2019SO3Wiley,hashim2020SE3Stochastic,moeini2016global}
are not suitable for pose estimation when a gyroscope fails and/or
a reliable GPS signal is not available. Consequently, multiple GPS-independent
navigation solutions that rely solely on angular velocity, landmark,
and inertial measurements have been proposed for the estimation of
attitude, position, and linear velocity of a vehicle, such as adaptive
Kalman filter \cite{aghili2016robust}, extended Kalman filter \cite{pesce2020radial},
invariant extended Kalman filter on the Lie Group of Extended Special Euclidean Group $\mathbb{SE}_{2}$ \cite{barrau2016invariant}, and
nonlinear stochastic filters on $\mathbb{SE}_{2}$ \cite{hashim2021Navigation,hashim2021gps}.
However, none of the above solutions account for gyroscope failure,
immediate replacement of which may prove challenging and expensive
\cite{Nasa2018Hubble}. As such, full observers that bypass measuring
angular and linear velocity, and provide accurate estimates of attitude,
angular velocity, position, and linear velocity are still lacking.

On the other hand, control of UAVs, in particular quadrotors and Vertical
Take-Off and Landing (VTOL)-UAVs, has drawn attention of the control
community in the recent years. Proposed solutions include backstepping
control \cite{roza2014class}, cascaded control \cite{su2011robust},
sliding mode control \cite{rios2018continuous,besnard2012quadrotor},
a hierarchical design procedure for the position control \cite{drouot2014hierarchical,hua2009control},
and others. The design of the above-mentioned controllers implies
the ready availability of accurate attitude, position, and angular
velocity which can be enabled only by high accuracy and precision
expensive sensors. Owing to large size, expensive sensors are unsuitable
for low-cost small-sized UAVs \cite{hashim2020SE3Stochastic}. Low-cost
UAVs are commonly equipped with low-cost sensors, such as an Inertial
Measurement Unit (IMU) and a vision unit (monocular or stereo camera)
\cite{hashim2021Navigation}. Note that a low-cost IMU provides noisy
angular velocity measurements \cite{hashim2019SO3Wiley}. Therefore,
integrating the above controllers with low-cost vision and IMU units
may lead to undesirable results \cite{hashim2021Navigation}. Alternatively,
the vehicle's orientation, position, and angular velocity can be obtained
by a combination of an Image-Based Visual Servoing (IBVS) algorithm
and an IMU. For instance, trajectory of a VTOL-UAV can be controlled
based on the information supplied by IBVS and an IMU \cite{zheng2016image,chen2019image,mokhtari2006feedback,lee2012autonomous}.
The aforementioned control architecture incorporates two loops, where
the inner loop controls the vehicle's orientation and angular velocity
employing IMU measurements, while the outer loop controls the position
using the thrust calculated based on the vision measurements. Nevertheless,
the state vector of most proposed IBVS solutions relies on Euler angles
which are subject to singularity, and therefore fail to represent
the attitude at certain configurations \cite{shuster1993survey,hashim2019AtiitudeSurvey}.
Furthermore, the solutions reported in \cite{zheng2016image,chen2019image,mokhtari2006feedback,lee2012autonomous}
are only locally stable. Considering the high nonlinearity of the
true motion dynamics of a VTOL-UAV, this paper proposes an observer-based
controller on the Lie Group that represents vehicle's orientation
on $\mathbb{SO}\left(3\right)$ providing a unique, global, and nonsingular
representation of the VTOL-UAV motion.

\paragraph*{Contributions}Motivated by the shortcomings of the existing
literature solutions and the high demand for observer-based controllers,
the contributions of this work are as follows:
\begin{itemize}
	\item[(1)] A nonlinear observer for attitude, angular velocity, position, and
	linear velocity that mimics the true motion dynamics of a VTOL-UAV
	is proposed on the Lie Group of $\mathbb{SE}_{2}\left(3\right)\times\mathbb{R}^{3}$.
	\item[(2)] The proposed observer operates based on measurements obtained from
	a vision unit without the need for angular and linear velocity measurements.
	\item[(3)] The closed loop error signals of the observer are shown to be almost
	globally exponentially stable.
	\item[(4)] A novel control law posed on the Lie Group is proposed. In combination
	with the proposed observer it forms an observer-based controller whose
	closed loop error signals are almost globally exponentially stable.
	\item[(5)] The proposed approach is continuous, and its discrete version is
	tested at a low sampling rate through simulation and experimentally.
\end{itemize}
The proposed observer-based controller allows for successful mission
completion even in case of gyroscope failure. In addition, it is suitable
for both GPS and GPS-denied applications. To the best of our knowledge,
this work is the first to present an observer-based controller on
the Lie Group that mimics the true VTOL-UAV motion dynamics and accurately
estimates attitude, angular velocity, position, and linear velocity.

\paragraph*{Structure}The rest of the paper is organized as follows:
Section \ref{sec:Preliminaries-and-Math} presents preliminaries;
Section \ref{sec:SE3_Problem-Formulation} formulates the problem;
Section \ref{sec:VTOL_Observer} proposes a novel nonlinear observer
for a VTOL-UAV; Section \ref{sec:VTOL_Controller} presents a VTOL-UAV
control strategy; Section \ref{sec:VTOL_Implementation} summarizes
the discrete implementation steps; The robustness of the proposed
approach is validated in Section \ref{sec:SE3_Simulations} through
simulation and experimental results; Finally, Section \ref{sec:SE3_Conclusion}
contains concluding remarks.

Table \ref{tab:Table-of-Notations2} provides some important notation
that will be used throughout the paper.

\begin{table}[t]
	\centering{}\caption{\label{tab:Table-of-Notations2}Nomenclature}
	\begin{tabular}{ll>{\raggedright}p{5.9cm}}
		\toprule 
		\addlinespace[0.1cm]
		$\left\{ \mathcal{B}\right\}$
		/ $\left\{ \mathcal{I}\right\}$ & : & 
		fixed body-frame / fixed inertial-frame\tabularnewline
		\addlinespace[0.1cm]
		$\mathbb{SO}\left(3\right)$ & : & Special Orthogonal Group of order 3\tabularnewline
		\addlinespace[0.1cm]
		$\mathfrak{so}\left(3\right)$ & : & Lie-algebra of $\mathbb{SO}\left(3\right)$\tabularnewline
		\addlinespace[0.1cm]
		$\mathbb{SE}\left(3\right)$ & : & Special Euclidean Group, $\mathbb{SE}\left(3\right)=\mathbb{SO}\left(3\right)\times\mathbb{R}^{3}$\tabularnewline
		\addlinespace[0.1cm]
		$\mathbb{SE}_{2}\left(3\right)$ & : & Extended $\mathbb{SE}\left(3\right)$, $\mathbb{SE}_{2}\left(3\right)=\mathbb{SE}\left(3\right)\times\mathbb{R}^{3}$\tabularnewline
		\addlinespace[0.1cm]
		$\mathbb{S}^{3}$ & : & \noindent Three-unit-sphere\tabularnewline
		\addlinespace[0.1cm]
		$\mathbb{R}^{n\times m}$ & : & $n$-by-$m$ real dimensional space\tabularnewline
		\addlinespace[0.1cm]
		\textcolor{red}{$R$,} $\hat{R}$, and $R_{d}$ & : & \textcolor{red}{True} (\textcolor{red}{unknown}), estimated, and desired attitude, $R,\hat{R},R_{d}\in\mathbb{SO}\left(3\right)$\tabularnewline
		\addlinespace[0.1cm]
		\textcolor{red}{$\Omega$}, $\hat{\Omega}$, and $\Omega_{d}$ & : & \textcolor{red}{True} (\textcolor{red}{unknown}), estimated, and desired angular velocity, $\Omega,\hat{\Omega},\Omega_{d}\in\mathbb{R}^{3}$\tabularnewline
		\addlinespace[0.1cm]
		\textcolor{red}{$P$}, $\hat{P}$, and $P_{d}$ & : & \textcolor{red}{True} (\textcolor{red}{unknown}), estimated, and desired position, $P,\hat{P},P_{d}\in\mathbb{R}^{3}$\tabularnewline
		\addlinespace[0.1cm]
		\textcolor{red}{$V$}, $\hat{V}$, and $V_{d}$ & : & \textcolor{red}{True} (\textcolor{red}{unknown}), estimated, and desired linear velocity, $V,\hat{V},V_{d}\in\mathbb{R}^{3}$\tabularnewline
		\addlinespace[0.1cm]
		\textcolor{red}{$X$} and $\hat{X}$ & : & \textcolor{red}{True} (\textcolor{red}{unknown}) and estimated navigation, $X,\hat{X}\in\mathbb{SE}_{2}\left(3\right)$\tabularnewline
		\addlinespace[0.1cm]
		$\mathcal{T}\in\mathbb{R}^{3}$ & : & Rotational torque input\tabularnewline
		\addlinespace[0.1cm]
		$\Im\in\mathbb{R}$ & : & Thrust magnitude input\tabularnewline
		\addlinespace[0.1cm]
		$y_{i}^{\mathcal{B}}\in\mathbb{R}^{3}$ & : & The $i$th body-frame vector measurement\tabularnewline
		\addlinespace[0.1cm]
		${\rm v}_{i}^{\mathcal{I}}\in\mathbb{R}^{3}$ & : & The $i$th inertial-frame observation\tabularnewline
		\addlinespace[0.1cm]
		$z_{j}^{\mathcal{B}}\in\mathbb{R}^{3}$ & : & Measured $j$th landmark at body-frame\tabularnewline
		\addlinespace[0.1cm]
		$p_{j}^{\mathcal{I}}\in\mathbb{R}^{3}$ & : & The $j$th landmark observation at inertial-frame\tabularnewline
		\addlinespace[0.1cm]
		$b_{\star}^{\mathcal{B}}\in\mathbb{R}^{3}$ & : & The $\star$th bias component of $y_{\star}^{\mathcal{B}}$ measurement\tabularnewline
		\addlinespace[0.1cm]
		$n_{\star}^{\mathcal{B}}\in\mathbb{R}^{3}$ & : & The $\star$th noise component of $y_{\star}^{\mathcal{B}}$ measurement\tabularnewline
		\addlinespace[0.1cm]
		$R_{y}\in\mathbb{SO}\left(3\right)$ & : & Reconstructed attitude\tabularnewline
		\addlinespace[0.1cm]
		$\tilde{R}_{o}\in\mathbb{SO}\left(3\right)$ & : & Attitude estimation error\tabularnewline
		\addlinespace[0.1cm]
		$\tilde{\Omega}_{o}\in\mathbb{R}^{3}$ & : & Angular velocity estimation error\tabularnewline
		\addlinespace[0.1cm]
		$\tilde{P}_{o}\in\mathbb{R}^{3}$ & : & Position estimation error\tabularnewline
		\addlinespace[0.1cm]
		$\tilde{V}_{o}\in\mathbb{R}^{3}$ & : & Linear velocity estimation error\tabularnewline
		\addlinespace[0.1cm]
		$\tilde{R}_{c}\in\mathbb{SO}\left(3\right)$ & : & Attitude control error\tabularnewline
		\addlinespace[0.1cm]
		$\tilde{\Omega}_{c}\in\mathbb{R}^{3}$ & : & Angular velocity control error\tabularnewline
		\addlinespace[0.1cm]
		$\tilde{P}_{c}\in\mathbb{R}^{3}$ & : & Position control error\tabularnewline
		\addlinespace[0.1cm]
		$\tilde{V}_{c}\in\mathbb{R}^{3}$ & : & Linear velocity control error\tabularnewline
		\addlinespace[0.1cm]
		$F\in\mathbb{R}^{3}$ & : & Intermediary control input\tabularnewline
		\addlinespace[0.1cm]
		$m$ and $J$  & : & Mass and inertia of the UAV, $m\in\mathbb{R}$ and $J\in\mathbb{R}^{3\times3}$\tabularnewline
		\addlinespace[0.1cm]
		$Q$, $\hat{Q}$, and $Q_{d}$ & : & True (unknown), estimated, and desired unit-quaternion vector, $Q,\hat{Q},Q_{d}\in\mathbb{S}^{3}$\tabularnewline
		\addlinespace[0.1cm]
		$\mathcal{R}_{Q}\in\mathbb{SO}\left(3\right)$ & : & Attitude representation obtained using unit-quaternion vector\tabularnewline
		\bottomrule
	\end{tabular}
\end{table}

\section{Preliminaries\label{sec:Preliminaries-and-Math}}

The set of real numbers, an $n$-by-$m$ real dimensional space, and
non-negative real numbers are represented by $\mathbb{R}$, $\mathbb{R}^{n\times m}$,
and $\mathbb{R}_{+}$, respectively. $||x||=\sqrt{x^{\top}x}$ refers
to an Euclidean norm of a vector $x\in\mathbb{R}^{n}$, while $||M||_{F}=\sqrt{{\rm Tr}\{MM^{*}\}}$
denotes the Frobenius norm of a matrix $M\in\mathbb{R}^{n\times m}$
where $*$ is the conjugate transpose. The set of eigenvalues of a
given matrix $M\in\mathbb{R}^{n\times n}$ is represented by $\lambda(M)=\{\lambda_{1},\lambda_{2},\ldots,\lambda_{n}\}$
with $\overline{\lambda}_{M}=\overline{\lambda}(M)$ and $\underline{\lambda}_{M}=\underline{\lambda}(M)$
being the set's maximum and minimum values, respectively. $0_{n\times m}$
represents an $n$-by-$m$ dimensional zero matrix, while $\mathbf{I}_{n}$
is an $n$-by-$n$ identity matrix. Consider a vehicle navigating
in 3D space with
\begin{itemize}
\item $\left\{ \mathcal{B}\right\} \triangleq\{e_{\mathcal{B}1},e_{\mathcal{B}2},e_{\mathcal{B}3}\}$
signifying the fixed body-frame attached to a vehicle and
\item $\left\{ \mathcal{I}\right\} \triangleq\{e_{1},e_{2},e_{3}\}$
representing the fixed inertial-frame.
\end{itemize}
The standard basis vectors
of $\mathbb{R}^{3}$ are denoted by $e_{1}:=[1,0,0]^{\top}$, $e_{2}:=[0,1,0]^{\top}$,
and $e_{3}:=[0,0,1]^{\top}$. Note that for $x\in \mathbb{R}^{n}$ the $m$th derivative of $x$ is defined by $x^{(m)}=dx^{m}/dt^{m}$.

\subsection{Lie Group of $\mathbb{SO}\left(3\right)$ and Properties}

The vehicle's orientation in 3D space is termed attitude commonly
represented as a rotation matrix in $\{\mathcal{B}\}$ defined by
$R\in\mathbb{SO}\left(3\right)\subset\mathbb{R}^{3\times3}$. The
notation $\mathbb{SO}\left(3\right)$ refers to the 3-dimensional
\textit{Special Orthogonal Group} defined by 
\[
\mathbb{SO}\left(3\right)=\left\{ \left.R\in\mathbb{R}^{3\times3}\right|RR^{\top}=R^{\top}R=\mathbf{I}_{3}\text{, }{\rm det}\left(R\right)=+1\right\} 
\]
Define $T_{R}\mathbb{SO}\left(3\right)\in\mathbb{R}^{3\times3}$ as
a tangent space of $\mathbb{SO}\left(3\right)$ at point $R\in\mathbb{SO}\left(3\right)$.
The \textit{Lie-algebra} of $\mathbb{SO}\left(3\right)$ is termed
$\mathfrak{so}\left(3\right)$ and follows the map $\left[\cdot\right]_{\times}:\mathbb{R}^{3}\rightarrow\mathfrak{so}\left(3\right)$
\begin{align*}
	\mathfrak{so}\left(3\right) & =\left\{ \left.\left[y\right]_{\times}\in\mathbb{R}^{3\times3}\right|\left[y\right]_{\times}^{\top}=-\left[y\right]_{\times}\right\} \\
	\left[y\right]_{\times} & =\left[\begin{array}{ccc}
		0 & -y_{3} & y_{2}\\
		y_{3} & 0 & -y_{1}\\
		-y_{2} & y_{1} & 0
	\end{array}\right]\in\mathfrak{so}\left(3\right),\hspace{1em}y=\left[\begin{array}{c}
		y_{1}\\
		y_{2}\\
		y_{3}
	\end{array}\right]
\end{align*}
with $[y]_{\times}$ being a skew symmetric matrix such that $[y]_{\times}z=y\times z$
for all $y,z\in\mathbb{R}^{3}$. The inverse mapping of $\left[\cdot\right]_{\times}$
to $\mathbb{R}^{3}$ is defined by $\mathbf{vex}:\mathfrak{so}\left(3\right)\rightarrow\mathbb{R}^{3}$
such that
\begin{equation}
	\mathbf{vex}([y]_{\times})=y,\hspace{1em}\forall y\in\mathbb{R}^{3}\label{eq:VTOL_VEX}
\end{equation}
The anti-symmetric projection operator $\boldsymbol{\mathcal{P}}_{a}$
on the $\mathfrak{so}\left(3\right)$ is defined by
\begin{align}
	\boldsymbol{\mathcal{P}}_{a}(A) & =\frac{1}{2}(A-A^{\top})\in\mathfrak{so}\left(3\right),\hspace{1em}\forall A\in\mathbb{R}^{3\times3}\label{eq:VTOL_Pa}\\
	\mathbf{vex}(\boldsymbol{\mathcal{P}}_{a}(A)) & =\frac{1}{2}[A_{32}-A_{23},A_{13}-A_{31},A_{21}-A_{12}]^{\top}\label{eq:VTOL_VEX_a}
\end{align}
where $A:=[A_{ij}]_{i,j=1,2,3}$. Define $||R||_{{\rm I}}$ as the
normalized Euclidean distance of $R\in\mathbb{SO}\left(3\right)$
such that
\begin{equation}
	||R||_{{\rm I}}=\frac{1}{4}{\rm Tr}\{\mathbf{I}_{3}-R\}\in\left[0,1\right]\label{eq:VTOL_Ecul_Dist}
\end{equation}
It is worth noting that $-1\leq{\rm Tr}\{R\}\leq3$ and $||R||_{{\rm I}}=\frac{1}{8}||\mathbf{I}_{3}-R||_{F}^{2}$
\cite{hashim2019AtiitudeSurvey}. Visit \cite{hashim2019SO3Wiley,hashim2019AtiitudeSurvey} for more information.

\subsection{$\mathbb{SE}\left(3\right)$, $\mathbb{SE}_{2}\left(3\right)$, and
	Tangent Space}

Let the vehicle's orientation, position, and linear velocity be denoted
as $R\in\mathbb{SO}\left(3\right)$, $P\in\mathbb{R}^{3}$, and $V\in\mathbb{R}^{3}$,
respectively. The \textit{Special Euclidean Group} is defined by $\mathbb{SE}\left(3\right):=\mathbb{SO}\left(3\right)\times\mathbb{R}^{3}\subset\mathbb{R}^{4\times4}$
where $T\in\mathbb{SE}\left(3\right)$ is a homogeneous transformation
matrix defined as follows:
\begin{equation}
	T=\left[\begin{array}{cc}
		R^{\top} & P\\
		0_{1\times3} & 1
	\end{array}\right],\hspace{1em}T^{-1}=\left[\begin{array}{cc}
		R & -RP\\
		0_{1\times3} & 1
	\end{array}\right]\label{eq:VTOL_T}
\end{equation}
visit \cite{hashim2020SE3Stochastic} for more information. $\mathbb{SE}_{2}\left(3\right)$
is the extended form of the \textit{Special Euclidean Group} introduced by \cite{barrau2016invariant} where
$\mathbb{SE}_{2}\left(3\right)=\mathbb{SO}\left(3\right)\times\mathbb{R}^{3}\times\mathbb{R}^{3}\subset\mathbb{R}^{5\times5}$
such that
\begin{equation}
	\mathbb{SE}_{2}\left(3\right)=\{\left.X\in\mathbb{R}^{5\times5}\right|R\in\mathbb{SO}\left(3\right),P,V\in\mathbb{R}^{3}\}\label{eq:VTOL_SE2_3}
\end{equation}
with $R$, $P$, and $V$ being vehicle's attitude, position and linear
velocity, respectively, and 
\begin{equation}
	X=\mathcal{N}(R^{\top},P,V)=\left[\begin{array}{ccc}
		R^{\top} & P & V\\
		0_{1\times3} & 1 & 0\\
		0_{1\times3} & 0 & 1
	\end{array}\right]\in\mathbb{SE}_{2}\left(3\right)\label{eq:VTOL_X}
\end{equation}
being its homogeneous navigation matrix (for more details see \cite{hashim2021Navigation,hashim2021gps}).
Note that 
\[
X^{-1}=\left[\begin{array}{ccc}
	R & -RP & -RV\\
	0_{1\times3} & 1 & 0\\
	0_{1\times3} & 0 & 1
\end{array}\right]\in\mathbb{SE}_{2}\left(3\right)
\]
The tangent space of $\mathbb{SE}_{2}\left(3\right)$ at point $X\in\mathbb{SE}_{2}\left(3\right)$
is defined by $T_{X}\mathbb{SE}_{2}\left(3\right)\in\mathbb{R}^{5\times5}$.
Define the submanifold $\mathcal{U}_{\mathcal{M}}=\mathfrak{so}\left(3\right)\times\mathbb{R}^{3}\times\mathbb{R}^{3}\times\mathbb{R}\subset\mathbb{R}^{5\times5}$
as
\begin{align}
	\mathcal{U}_{\mathcal{M}} & =\left\{ \left.u(\text{[\ensuremath{\Omega\text{\ensuremath{]_{\times}}}}},V,a,\kappa)\right|[\Omega\text{\ensuremath{]_{\times}}}\in\mathfrak{so}\left(3\right),V,a\in\mathbb{R}^{3},\kappa\in\mathbb{R}\right\} \nonumber \\
	& U=u([\Omega\text{\ensuremath{]_{\times}}},V,a,\kappa)=\left[\begin{array}{ccc}
		[\Omega\text{\ensuremath{]_{\times}}} & V & a\\
		0_{1\times3} & 0 & 0\\
		0_{1\times3} & \kappa & 0
	\end{array}\right]\in\mathcal{U}_{\mathcal{M}}\label{eq:VTOL_Mu}
\end{align}

\subsection{Unit-quaternion\label{subsec:Unit-quaternion}}

\noindent Let us define a set of 3-unit-sphere 
\[
\mathbb{S}^{3}=\{\left.Q=[q_{0},q^{\top}]^{\top}\in\mathbb{R}^{4}\right|||Q||=\sqrt{q_{0}^{2}+q^{\top}q}=1\}
\]
where $q_{0}\in\mathbb{R}$ and $q\in\mathbb{R}^{3}$. Consider the
inverse of $Q\in\mathbb{S}^{3}$ to be $Q^{-1}=[\begin{array}{cc}
	q_{0} & -q^{\top}\end{array}]^{\top}\in\mathbb{S}^{3}$. Let $\odot$ be a quaternion product. For $Q_{1}=[\begin{array}{cc}
	q_{01} & q_{1}^{\top}\end{array}]^{\top}\in\mathbb{S}^{3}$ and $Q_{2}=[\begin{array}{cc}
	q_{02} & q_{2}^{\top}\end{array}]^{\top}\in\mathbb{S}^{3}$, one has
\[
Q_{1}\odot Q_{2}=\left[\begin{array}{c}
	q_{01}q_{02}-q_{1}^{\top}q_{2}\\
	q_{01}q_{2}+q_{02}q_{1}+[q_{1}]_{\times}q_{2}
\end{array}\right]
\]
The mapping from $\mathbb{S}^{3}$ to $\mathbb{SO}\left(3\right)$
is given by
\begin{align}
	\mathcal{R}_{Q} & =(q_{0}^{2}-||q||^{2})\mathbf{I}_{3}+2qq^{\top}-2q_{0}\left[q\right]_{\times}\in\mathbb{SO}\left(3\right)\label{eq:NAV_Append_SO3}
\end{align}
The following two identities will be utilized in the subsequent derivations:
\begin{align}
	[Ry]_{\times}= & R[y]_{\times}R^{\top},\hspace{1em}y\in{\rm \mathbb{R}}^{3},R\in\mathbb{SO}\left(3\right)\label{eq:VTOL_R_Identity1}\\
	{\rm Tr}\{A[y]_{\times}\}= & {\rm Tr}\{\boldsymbol{\mathcal{P}}_{a}(A)[y]_{\times}\},\hspace{1em}y\in{\rm \mathbb{R}}^{3},A\in\mathbb{R}^{3\times3}\nonumber \\
	= & -2\mathbf{vex}(\boldsymbol{\mathcal{P}}_{a}(A))^{\top}y\label{eq:VTOL_R_Identity2}
\end{align}

\section{Problem Formulation and Measurements\label{sec:SE3_Problem-Formulation}}

Consider a UAV navigating in 3D space. Let its true attitude, angular
velocity, position, and linear velocity be unknown and described by
$R\in\mathbb{SO}\left(3\right)$, $\Omega\in\mathbb{R}^{3}$, $P\in\mathbb{R}^{3}$,
and $V\in\mathbb{R}^{3}$, respectively. Note that while $R,\Omega\in\{\mathcal{B}\}$
are defined with respect to the body-frame, $P,V\in\{\mathcal{I}\}$
are represented with respect to the inertial-frame. The dynamical
equations of a VTOL-UAV are given by
\begin{align}
	\text{Rotation} & \begin{cases}
		\dot{R} & =-\left[\Omega\right]_{\times}R\\
		J\dot{\Omega} & =\left[J\Omega\right]_{\times}\Omega+\mathcal{T}
	\end{cases},\hspace{1em}R,\Omega\in\{\mathcal{B}\}\label{eq:VTOL_Rotation}\\
	\text{Translation} & \begin{cases}
		\dot{P} & =V\\
		\dot{V} & =ge_{3}-\frac{\Im}{m}R^{\top}e_{3}
	\end{cases},\hspace{1em}P,V\in\{\mathcal{I}\}\label{eq:VTOL_Translation}
\end{align}
with $\mathcal{T}\in\mathbb{R}^{3}$ being the external torque input,
$\Im\in\mathbb{R}$ being the thrust magnitude input in the direction
of $e_{\mathcal{B}3}$ (see Fig. \ref{fig:VTOL}), and $J\in\mathbb{R}^{3\times3}$
being a constant symmetric positive definite inertia matrix. $e_{3}=[0,0,1]^{\top}$,
$m$, and $g$ denote standard basis vector, vehicle's constant mass,
and gravitational acceleration, respectively. Note that $J,\mathcal{T}\in\{\mathcal{B}\}$.
The set in \eqref{eq:VTOL_Rotation} describes the true VTOL-UAV rotational
dynamics, while the set in \eqref{eq:VTOL_Translation} describes
the true VTOL-UAV translational dynamics. It is apparent that the
nonlinear attitude dynamics in \eqref{eq:VTOL_Rotation} follow the
map $\mathbb{SO}\left(3\right)\times\mathfrak{so}\left(3\right)\rightarrow T_{R}\mathbb{SO}\left(3\right)$.
The nonlinear dynamics in \eqref{eq:VTOL_Rotation} and \eqref{eq:VTOL_Translation}
can be rewritten compactly as follows: 
\begin{equation}
	\begin{cases}
		\dot{X} & =XU-\mathcal{G}X\\
		J\dot{\Omega} & =\left[J\Omega\right]_{\times}\Omega+\mathcal{T}
	\end{cases}\label{eq:VTOL_Compact}
\end{equation}
where the navigation matrix $X\in\mathbb{SE}_{2}\left(3\right)$ is
as defined in \eqref{eq:VTOL_X}, $U=\underbrace{\left[\begin{array}{ccc}
		[\Omega\text{\ensuremath{]_{\times}}} & 0_{3\times1} & -\frac{\Im}{m}e_{3}\\
		0_{1\times3} & 0 & 0\\
		0_{1\times3} & 1 & 0
	\end{array}\right]}_{u([\Omega\text{\ensuremath{]_{\times}}},0_{3\times1},-\frac{\Im}{m}e_{3},1)}\in\mathcal{U}_{\mathcal{M}},$ and $\mathcal{G}=\underbrace{\left[\begin{array}{ccc}
		0_{3\times3} & 0_{3\times1} & -ge_{3}\\
		0_{1\times3} & 0 & 0\\
		0_{1\times3} & 1 & 0
	\end{array}\right]}_{u(0_{3\times3},0_{3\times1},-ge_{3},1)}\in\mathcal{U}_{\mathcal{M}}$, see \eqref{eq:VTOL_Mu}. The nonlinear dynamics in \eqref{eq:VTOL_Compact}
follow the map $\mathbb{SE}_{2}\left(3\right)\times\mathcal{U}_{\mathcal{M}}\rightarrow T_{X}\mathbb{SE}_{2}\left(3\right)$
with $\dot{X}\in T_{X}\mathbb{SE}_{2}\left(3\right)$. Fig. \ref{fig:VTOL}
schematically depicts the VTOL-UAV estimation and tracking control
problem.
\begin{figure}
	\centering{}\includegraphics[scale=0.49]{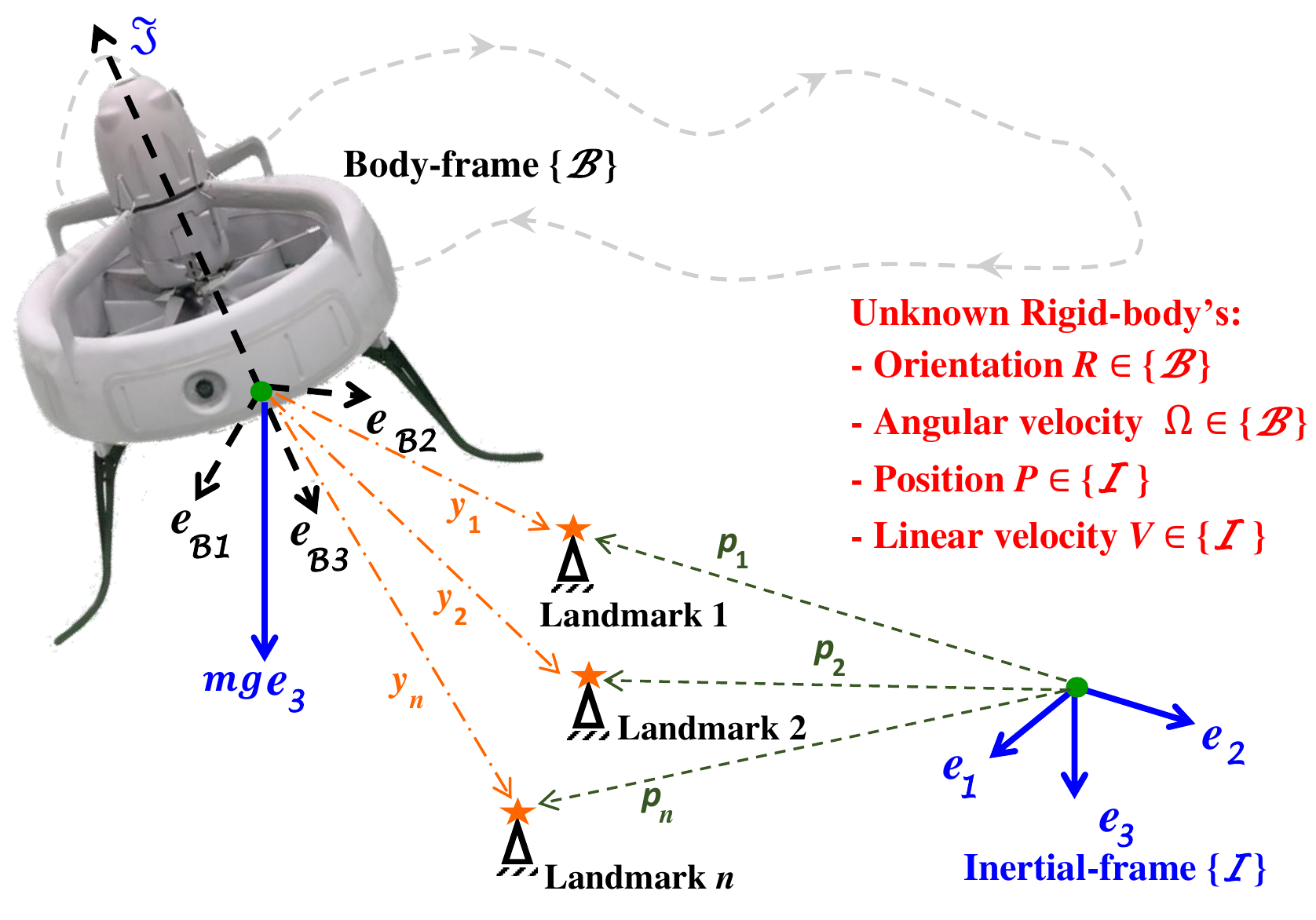}\caption{VTOL-UAV estimation and tracking control problem.}
	\label{fig:VTOL}
\end{figure}

\subsection{Inertial Measurements}

Considering the fact that the true VTOL-UAV motion parameters $R$,
$\Omega$, $P$, and $V$ are unknown, the estimation process requires
sensor measurements. For simplicity, let the superscripts $\mathcal{I}$
and $\mathcal{B}$ indicate association with inertial-frame and body-frame,
respectively. Given a set of observations in $\{\mathcal{I}\}$ and
the corresponding $\{\mathcal{B}\}$ measurements, the orientation
of a vehicle can be obtained as follows \cite{hashim2019SO3Wiley,zlotnik2016exponential}:
\begin{align}
	y_{i}^{\mathcal{B}} & =R{\rm v}_{i}^{\mathcal{I}}+b_{i}^{\mathcal{B}}+n_{i}^{\mathcal{B}}\in\mathbb{R}^{3}\label{eq:VTOL_VR}
\end{align}
where $y_{i}^{\mathcal{B}}$ refers to vector measurements, ${\rm v}_{i}^{\mathcal{I}}$
refers to a known observation, $b_{i}^{\mathcal{B}}$ denotes unknown
constant bias, and $n_{i}^{\mathcal{B}}$ describes unknown random
noise associated with the $i$th measurement for all $i=1,2,\ldots,N_{1}$.
Note that \eqref{eq:VTOL_VR} represents a typical low-cost IMU module
(e.g., magnetometer and accelerometer). Moreover, attitude and position
can be obtained via a vision unit (monocular or stereo camera) using
a group of known landmarks in $\{\mathcal{I}\}$ and their measurements
in $\{\mathcal{B}\}$ where the $j$th measurement is defined by \cite{hashim2020SE3Stochastic}
\begin{align}
	z_{j}^{\mathcal{B}} & =R(p_{j}^{\mathcal{I}}-P)+b_{j}^{\mathcal{B}}+n_{j}^{\mathcal{B}}\in\mathbb{R}^{3}\label{eq:VTOL_VRP}
\end{align}
with $p_{j}^{\mathcal{I}}$ denoting a known landmark, $b_{j}^{\mathcal{B}}$
denoting unknown bias (constant), and $n_{j}^{\mathcal{B}}$ denoting
unknown random noise for all $j=1,2,\ldots,N_{2}$. Define $s_{j}$
as the $j$th measurement sensor confidence level, and let $s_{c}=\sum_{j=1}^{N_{2}}s_{j}$.
Define the landmark weighted geometric center (observations and measurements)
as
\begin{equation}
	p_{c}=\frac{1}{s_{c}}\sum_{j=1}^{N_2}s_{j}p_{j}^{\mathcal{I}},\hspace{1em}z_{c}=\frac{1}{s_{c}}\sum_{j=1}^{N_{2}}s_{j}z_{j}^{\mathcal{B}}\label{eq:VTOL_R_Weighted}
\end{equation}
The measurement in \eqref{eq:VTOL_VR} can be reformulated in terms
of the homogeneous transformation matrix $T\in\mathbb{SE}\left(3\right)$
in \eqref{eq:VTOL_T} as $\overline{y}_{i}^{\mathcal{B}}=T^{-1}\overline{{\rm v}}_{i}^{\mathcal{I}}+\overline{b}_{i}^{\mathcal{B}}+\overline{n}_{i}^{\mathcal{B}}\in\mathbb{R}^{4}$
where $\overline{y}_{i}^{\mathcal{B}}=[(y_{i}^{\mathcal{B}})^{\top},0]^{\top}$,
$\overline{{\rm v}}_{i}^{\mathcal{I}}=[({\rm v}_{i}^{\mathcal{I}})^{\top},0]^{\top}$,
$\overline{b}_{i}^{\mathcal{B}}=[(b_{i}^{\mathcal{B}})^{\top},0]^{\top}$,
and $\overline{n}_{i}^{\mathcal{B}}=[(n_{i}^{\mathcal{B}})^{\top},0]^{\top}$.
Likewise, the measurement in \eqref{eq:VTOL_VRP} can be described
with respect to $T\in\mathbb{SE}\left(3\right)$ as $\overline{z}_{j}^{\mathcal{B}}=T^{-1}\overline{p}_{j}^{\mathcal{I}}+\overline{b}_{j}^{\mathcal{B}}+\overline{n}_{j}^{\mathcal{B}}\in\mathbb{R}^{4}$
where $\overline{z}_{j}^{\mathcal{B}}=[(z_{j}^{\mathcal{B}})^{\top},1]^{\top}$,
$\overline{p}_{j}^{\mathcal{I}}=[(p_{j}^{\mathcal{I}})^{\top},1]^{\top}$,
$\overline{b}_{j}^{\mathcal{B}}=[(b_{j}^{\mathcal{B}})^{\top},0]^{\top}$,
and $\overline{n}_{j}^{\mathcal{B}}=[(n_{j}^{\mathcal{B}})^{\top},0]^{\top}$.

\begin{assum}\label{Assum:VTOL_1Landmark}(Pose observability) The
	pose of a vehicle $T\in\mathbb{SE}\left(3\right)$ can be obtained
	if one of the following three conditions is met:
	\begin{enumerate}
		\item[A1.] observations in $\{\mathcal{I}\}$ and the associated $\{\mathcal{B}\}$
		measurements of a minimum of one landmark as in \eqref{eq:VTOL_VRP}
		and two different inertial vectors as in \eqref{eq:VTOL_VR} are non-collinear.
		\item[A2.] observations in $\{\mathcal{I}\}$ and the associated $\{\mathcal{B}\}$
		measurements of a minimum of two different landmarks as in \eqref{eq:VTOL_VRP}
		and one inertial vector as in \eqref{eq:VTOL_VR} are non-collinear.
		\item[A3.] observations in $\{\mathcal{I}\}$ and the associated $\{\mathcal{B}\}$
		measurements of a minimum of three different landmarks as in \eqref{eq:VTOL_VRP}
		are non-collinear.
	\end{enumerate}
\end{assum}

It is worth mentioning that Assumption \ref{Assum:VTOL_1Landmark}
is standard for pose filtering \cite{hashim2020SE3Stochastic,moeini2016global}.

\begin{assum}\label{Assum:P_Om_desired}Let $P_{d}$ denote the desired
	position of a VTOL-UAV with $\dot{P}_{d}=V_{d}$, $\ddot{P}_{d}$,
	$P_{d}^{(3)}$, and $P_{d}^{(4)}$ being its first, second, third,
	and fourth derivatives, respectively. Also, let $\Omega_{d}$ and
	$\dot{\Omega}_{d}$ be the desired angular velocity and its rate of
	change, respectively. $P_{d}$, $V_{d}$, $\ddot{P}_{d}$, $P_{d}^{(3)}$,
	$P_{d}^{(4)}$, $\Omega_{d}$, and $\dot{\Omega}_{d}$ are assumed
	to be uniformly upper bounded in time.\end{assum}
\begin{lem}
	\label{Lemm:vex_RI}Let $R\in\mathbb{SO}\left(3\right)$, and consider
	the definitions in \eqref{eq:VTOL_VEX_a} and \eqref{eq:VTOL_Ecul_Dist}.
	Consequently, the following holds:
	\begin{equation}
		||\mathbf{vex}(\boldsymbol{\mathcal{P}}_{a}(R))||^{2}=4(1-||R||_{{\rm I}})||R||_{{\rm I}}\label{eq:VTOL_lemm_vexRI}
	\end{equation}
	\begin{proof}See the Appendix in \cite{hashim2019SO3Wiley}.\end{proof}
\end{lem}
\begin{defn}
	\label{def:Unstable-set} Define the following non-attractive and
	forward invariant unstable set $\mathcal{U}_{s}\subseteq\mathbb{SO}\left(3\right)$:
	\begin{equation}
		\mathcal{U}_{s}=\{\left.R(0)\in\mathbb{SO}\left(3\right)\right|{\rm Tr}\{R(0)\}=-1\}\label{eq:SO3_PPF_STCH_SET}
	\end{equation}
	where $R(0)\in\mathcal{U}_{s}$ in one of the following three cases:
	$R(0)={\rm diag}(-1,-1,1)$, $R(0)={\rm diag}(-1,1,-1)$, or $R(0)={\rm diag}(1,-1,-1)$.
\end{defn}
\begin{figure*}
	\centering{}\includegraphics[scale=0.4]{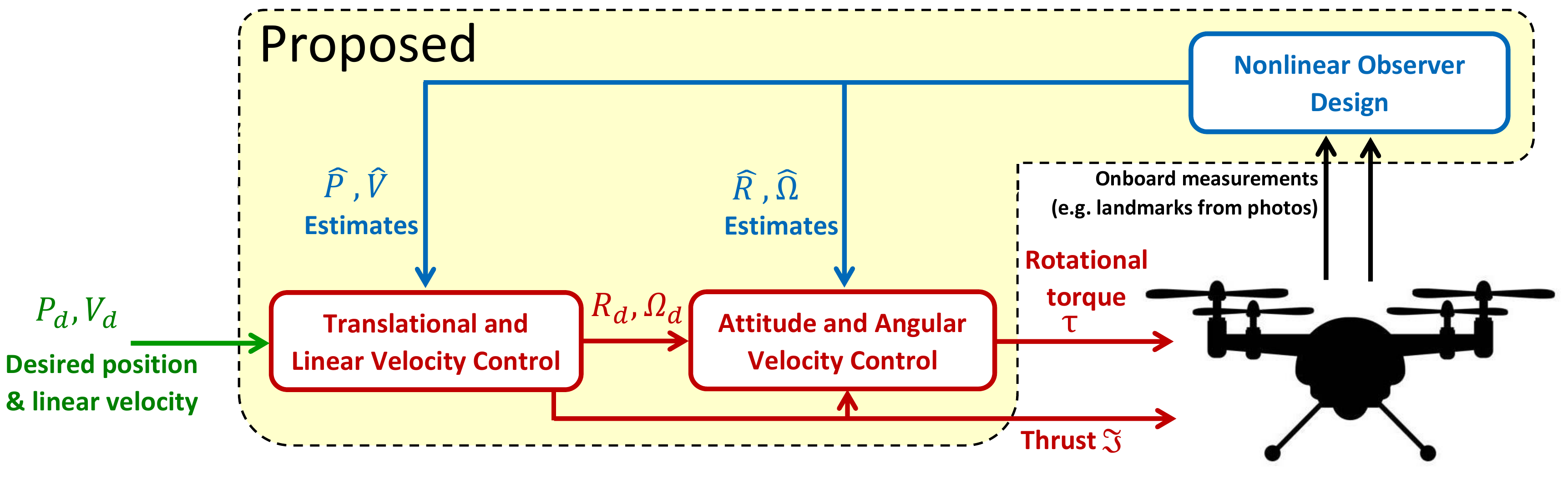}\caption{Graphical summery of the proposed observer-based control solution
		for a VTOL-UAV.}
	\label{fig:VTOL-1}
\end{figure*}

The goal of this work is two-fold and includes proposing a nonlinear
observer and a controller which are strongly coupled and designed
to operate as a module. Our first objective (Section \ref{sec:VTOL_Observer})
is to design a nonlinear observer able to estimate attitude ($\hat{R})$,
angular velocity ($\hat{\Omega})$, position ($\hat{P})$, and linear
velocity ($\hat{V})$ of a UAVs in six degrees of freedom (6 DoF)
using onboard sensor measurements. Herein, measurements refer to feature
information extracted from photos taken by a monocular or stereo camera.
The proposed solution does not require an IMU or a GPS, but can be integrated
with IMU and GPS measurements if needed. The second objective (Section \ref{sec:VTOL_Controller})
is to design a controller that uses components estimated by the observer
to generate rotational torque ($\mathcal{T}$) and thrust ($\Im$)
necessary to control the VTOL-UAV with respect to a desired position
and linear velocity trajectories. Fig. \ref{fig:VTOL-1} graphically
illustrates the research objective of this work.

\section{Nonlinear Observer Design on Lie Group \label{sec:VTOL_Observer}}

Define $P_{y}\in\mathbb{R}^{3}$ as a reconstructed position and $R_{y}\in\mathbb{SO}\left(3\right)$
as a reconstructed attitude of a VTOL-UAV (for more details visit
\cite{hashim2020SE3Stochastic,hashim2019SO3Wiley}). The concept of
reconstruction is detailed in Section \ref{sec:SE3_Simulations}%
\begin{comment}
	(Eq. \eqref{eq:SVD} and \eqref{eq:Rec})
\end{comment}
. For the sake of stability analysis of the observer design, it is
considered that the reconstructed components $P_{y}$ and $R_{y}$
are close to the true components $P$ and $R$. In the implementation,
on the contrary, the observer is tested against a high level of uncertainties
corrupting $P_{y}$ and $R_{y}$. Define $\hat{R}\in\mathbb{SO}\left(3\right)$,
$\hat{\Omega}\in\mathbb{R}^{3}$, $\hat{P}\in\mathbb{R}^{3}$, and
$\hat{V}\in\mathbb{R}^{3}$ as the estimates of the true attitude
$R\in\mathbb{SO}\left(3\right)$, angular velocity $\Omega\in\mathbb{R}^{3}$,
position $P\in\mathbb{R}^{3}$, and linear velocity $V\in\mathbb{R}^{3}$,
respectively. Define the errors between the estimated and the true
values of attitude, angular velocity, position, and linear velocity
as
\begin{align}
	\tilde{R}_{o}= & R\hat{R}^{\top}\label{eq:VTOL_Rerr}\\
	\tilde{\Omega}_{o}= & \Omega-\tilde{R}_{o}\hat{\Omega}\label{eq:VTOL_Omerr}\\
	\tilde{P}_{o}= & P-\hat{P}\label{eq:VTOL_Perr}\\
	\tilde{V}_{o}= & V-\hat{V}\label{eq:VTOL_Verr}
\end{align}
where $\tilde{R}_{o}\in\mathbb{SO}\left(3\right)$, $\tilde{\Omega}_{o}\in\mathbb{R}^{3}$,
$\tilde{P}_{o}\in\mathbb{R}^{3}$, and $\tilde{V}_{o}\in\mathbb{R}^{3}$.
Consider the nonlinear dynamics $\dot{X}=XU-\mathcal{G}X\in T_{X}\mathbb{SE}_{2}\left(3\right)$
in \eqref{eq:VTOL_Compact} where $X\in\mathbb{SE}_{2}\left(3\right)$
and $U,\mathcal{G}\in\mathcal{U}_{\mathcal{M}}$ such that $\mathbb{SE}_{2}\left(3\right)\times\mathcal{U}_{\mathcal{M}}\rightarrow T_{X}\mathbb{SE}_{2}\left(3\right)$.
The objective of this Section is to propose a nonlinear observer on
the Lie group of $\mathbb{SE}_{2}\left(3\right)$ that mimics the
nonlinear dynamics in \eqref{eq:VTOL_Compact} with $\hat{X}\in\mathbb{SE}_{2}\left(3\right)$
being the estimate of $X$ such that $\dot{\hat{X}}\in T_{\hat{X}}\mathbb{SE}_{2}\left(3\right)$.
The proposed observer aims to drive 
\begin{align*}
	\hat{R} & \rightarrow R\\
	\hat{\Omega} & \rightarrow\Omega\\
	\hat{P} & \rightarrow P\\
	\hat{V} & \rightarrow V
\end{align*}
with $\lim_{t\rightarrow\infty}\tilde{R}_{o}=\mathbf{I}_{3}$ and
$\lim_{t\rightarrow\infty}\tilde{\Omega}_{o}=\lim_{t\rightarrow\infty}\tilde{P}_{o}=\lim_{t\rightarrow\infty}\tilde{V}_{o}=0_{3\times1}$.
Let us propose a nonlinear observer for VTOL-UAV on Lie group in a
compact form:
\begin{equation}
	\begin{cases}
		\dot{\hat{X}} & =\hat{X}\hat{U}-W\hat{X}\\
		\hat{J}\dot{\hat{\Omega}} & =[\hat{J}\hat{\Omega}]_{\times}\hat{\Omega}+\hat{\mathcal{T}}-\hat{J}[\hat{\Omega}]_{\times}\hat{R}w_{\Omega}+w_{o}
	\end{cases}\label{eq:VTOL_ObsvCompact}
\end{equation}
with the following set of correction factors:
\begin{equation}
	\begin{cases}
		w_{o} & =-\gamma_{o}\tilde{R}_{o}^{\top}\mathbf{vex}(\boldsymbol{\mathcal{P}}_{a}(\tilde{R}_{o})),\hspace{1em}\tilde{R}_{o}=R_{y}\hat{R}^{\top}\\
		w_{\Omega} & =k_{o1}R_{y}^{\top}\mathbf{vex}(\boldsymbol{\mathcal{P}}_{a}(\tilde{R}_{o}))\\
		w_{V} & =-\left[w_{\Omega}\right]_{\times}\hat{P}-k_{o2}\tilde{P}_{o},\hspace{1em}\tilde{P}_{o}=P_{y}-\hat{P}\\
		w_{a} & =-\frac{\Im}{m}\hat{R}^{\top}(\mathbf{I}_{3}-\tilde{R}_{o}^{\top})e_{3}-ge_{3}-\left[w_{\Omega}\right]_{\times}\hat{V}-k_{o3}\tilde{P}_{o}
	\end{cases}\label{eq:VTOL_Wcorr}
\end{equation}
such that 
\[
\hat{X}=\mathcal{N}(\hat{R}^{\top},\hat{P},\hat{V})=\left[\begin{array}{ccc}
	\hat{R}^{\top} & \hat{P} & \hat{V}\\
	0_{1\times3} & 1 & 0\\
	0_{1\times3} & 0 & 1
\end{array}\right]\in\mathbb{SE}_{2}\left(3\right)
\]
in accordance with the map in \eqref{eq:VTOL_X}, 
\[
\hat{U}=\underbrace{\left[\begin{array}{ccc}
		[\hat{\Omega}\text{\ensuremath{]_{\times}}} & 0_{3\times1} & -\frac{\Im}{m}e_{3}\\
		0_{1\times3} & 0 & 0\\
		0_{1\times3} & 1 & 0
	\end{array}\right]}_{u([\hat{\Omega}\text{\ensuremath{]_{\times}}},0_{3\times1},-\frac{\Im}{m}e_{3},1)}\in\mathcal{U}_{m}
\]
and 
\[
W=u([w_{\Omega}]_{\times},w_{V},w_{a},1)=\left[\begin{array}{ccc}
	[w_{\Omega}]_{\times} & w_{V} & w_{a}\\
	0_{1\times3} & 0 & 0\\
	0_{1\times3} & 1 & 0
\end{array}\right]\in\mathcal{U}_{m}
\]
as per the map in \eqref{eq:VTOL_Mu}, $J$, $m$, and $g$ denote
vehicle's inertia matrix, mass, and gravitational acceleration, respectively,
$\hat{\mathcal{T}}=\tilde{R}_{o}^{\top}\mathcal{T}$ and $\hat{J}=\tilde{R}_{o}^{\top}J\tilde{R}_{o}$
denote the rotational torque input and the inertia matrix, respectively,
and $\gamma_{o}$, $k_{o1}$, $k_{o2}$, and $k_{o3}$ stand for strictly
positive constants. It becomes apparent that $\dot{\hat{X}}\in T_{\hat{X}}\mathbb{SE}_{2}\left(3\right)$.
The detailed representation of the novel nonlinear observer in \eqref{eq:VTOL_ObsvCompact}
is as follows:
\begin{align}
	\dot{\hat{R}} & =\hat{R}[w_{\Omega}]_{\times}-[\hat{\Omega}]_{\times}\hat{R}\label{eq:VTOL_Restdot}\\
	\hat{J}\dot{\hat{\Omega}} & =[\hat{J}\hat{\Omega}]_{\times}\hat{\Omega}+\hat{\mathcal{T}}-\hat{J}[\hat{\Omega}]_{\times}\hat{R}w_{\Omega}+w_{o}\label{eq:VTOL_Omestdot}\\
	\dot{\hat{P}} & =\hat{V}-[w_{\Omega}]_{\times}\hat{P}-w_{V}\label{eq:VTOL_Pestdot}\\
	\dot{\hat{V}} & =-\frac{\Im}{m}\hat{R}^{\top}e_{3}-[w_{\Omega}]_{\times}\hat{V}-w_{a}\label{eq:VTOL_Vestdot}
\end{align}

\begin{thm}
	\label{thm:Theorem1}Consider VTOL-UAV dynamics in \eqref{eq:VTOL_Rotation}
	and \eqref{eq:VTOL_Translation}, and suppose that Assumption \ref{Assum:VTOL_1Landmark}
	is satisfied. Couple the inertial and landmark measurements $y_{i}^{\mathcal{B}}=R{\rm v}_{i}^{\mathcal{I}}$
	and $z_{j}^{\mathcal{B}}=R(p_{j}^{\mathcal{I}}-P)$ for all $i=1,2,\ldots,N_{1}$
	and $j=1,2,\ldots,N_{2}$ with the observer in \eqref{eq:VTOL_ObsvCompact}
	and the correction factors in \eqref{eq:VTOL_Wcorr}. Let $\gamma_{o}$,
	$k_{o1}$, $k_{o2}$, and $k_{o3}$ be positive constants and $\tilde{R}_{o}(0)\notin\mathcal{U}_{s}$
	(see Definition \ref{def:Unstable-set}). Define the set:
	\begin{align}
		\mathcal{S}_{o}= & \{(\tilde{R}_{o},\tilde{\Omega}_{o},\tilde{P}_{o},\tilde{V}_{o})\in\mathbb{SO}\left(3\right)\times\mathbb{R}^{3}\times\mathbb{R}^{3}\times\mathbb{R}^{3}\,|\nonumber \\
		& \hspace{6em}\tilde{R}_{o}=\mathbf{I}_{3},\tilde{\Omega}_{o}=\tilde{P}_{o}=\tilde{V}_{o}=0_{3\times1}\}\label{eq:VTOL_Set1}
	\end{align}
	Then, the set $\mathcal{S}_{o}$ is uniformly almost globally exponentially
	stable.
\end{thm}
\begin{proof}It is stated in Theorem \ref{thm:Theorem1} that error
	signals converge exponentially to the equilibrium point from any initial
	condition except for the three repeller orientations defined in Definition
	\ref{def:Unstable-set} ($\tilde{R}_{o}(0)\notin\mathcal{U}_{s}$).
	Consider the attitude error $\tilde{R}_{o}=R\hat{R}^{\top}$ in \eqref{eq:VTOL_Rerr}.
	From \eqref{eq:VTOL_Rotation} and \eqref{eq:VTOL_Restdot}, one has
	\begin{align}
		\dot{\tilde{R}}_{o} & =\dot{R}\hat{R}^{\top}+R\dot{\hat{R}}^{\top}\nonumber \\
		& =-\left[\Omega\right]_{\times}R\hat{R}^{\top}+R(-[w_{\Omega}]_{\times}\hat{R}^{\top}+\hat{R}^{\top}[\hat{\Omega}]_{\times})\nonumber \\
		& =-\left[\Omega\right]_{\times}\tilde{R}_{o}+\tilde{R}_{o}[\hat{\Omega}]_{\times}-[Rw_{\Omega}]_{\times}\tilde{R}_{o}\nonumber \\
		& =-[\Omega-\tilde{R}_{o}\hat{\Omega}+Rw_{\Omega}]_{\times}\tilde{R}_{o}=-[\tilde{\Omega}_{o}+Rw_{\Omega}]_{\times}\tilde{R}_{o}\label{eq:VTOL_Rerr_dot}
	\end{align}
	where $\tilde{\Omega}_{o}$ is defined in \eqref{eq:VTOL_Omerr}.
	Using \eqref{eq:VTOL_Ecul_Dist} and \eqref{eq:VTOL_Rerr_dot}, one
	obtains \cite{hashim2019SO3Wiley}
	\begin{align}
		||\dot{\tilde{R}}_{o}||_{{\rm I}}= & -\frac{1}{4}{\rm Tr}\{\dot{\tilde{R}}_{o}\}=\frac{1}{4}{\rm Tr}\{[\tilde{\Omega}_{o}+Rw_{\Omega}]_{\times}\boldsymbol{\mathcal{P}}_{a}(\tilde{R}_{o})\}\nonumber \\
		= & -\frac{1}{2}\mathbf{vex}(\boldsymbol{\mathcal{P}}_{a}(\tilde{R}_{o}))^{\top}(\tilde{\Omega}_{o}+Rw_{\Omega})\label{eq:VTOL_RerrI_dot}
	\end{align}
	From \eqref{eq:VTOL_Rotation}, \eqref{eq:VTOL_Omerr}, \eqref{eq:VTOL_Omestdot},
	and \eqref{eq:VTOL_Rerr_dot}, one finds
	\begin{align}
		J\dot{\tilde{\Omega}}_{o}= & J\dot{\Omega}-J\dot{\tilde{R}}_{o}\hat{\Omega}-J\tilde{R}_{o}\dot{\hat{\Omega}}\nonumber \\
		= & [J\Omega]_{\times}\Omega+\mathcal{T}+J[\tilde{\Omega}_{o}+Rw_{\Omega}]_{\times}\tilde{R}_{o}\hat{\Omega}-J\tilde{R}_{o}\dot{\hat{\Omega}}\nonumber \\
		= & [J\Omega]_{\times}\Omega+J[\tilde{\Omega}_{o}]_{\times}\tilde{R}_{o}\hat{\Omega}-[J\tilde{R}_{o}\hat{\Omega}]_{\times}\tilde{R}_{o}\hat{\Omega}\nonumber \\
		& \mathcal{T}+[J\tilde{R}_{o}\hat{\Omega}]_{\times}\tilde{R}_{o}\hat{\Omega}+J[Rw_{\Omega}]_{\times}\tilde{R}_{o}\hat{\Omega}-\tilde{R}_{o}\hat{J}\dot{\hat{\Omega}}\nonumber \\
		= & S(\Omega)\tilde{\Omega}_{o}-[J\tilde{\Omega}_{o}]_{\times}\tilde{\Omega}_{o}-\tilde{R}_{o}w_{o}\label{eq:VTOL_Omerr_dot}
	\end{align}
	such that
	\begin{align}
		& [J\Omega]_{\times}\Omega+J[\tilde{\Omega}_{o}]_{\times}\tilde{R}_{o}\hat{\Omega}-[J\tilde{R}_{o}\hat{\Omega}]_{\times}\tilde{R}_{o}\hat{\Omega}\nonumber \\
		& =[J\Omega]_{\times}\Omega+(J[\tilde{\Omega}_{o}]_{\times}-[J\Omega]_{\times}+[J\tilde{\Omega}_{o}]_{\times})\tilde{R}_{o}\hat{\Omega}\nonumber \\
		& =[J\Omega]_{\times}\tilde{\Omega}_{o}-J[\Omega]_{\times}\tilde{\Omega}_{o}-[\Omega]_{\times}J\tilde{\Omega}_{o}-[J\tilde{\Omega}_{o}]_{\times}\tilde{\Omega}_{o}\nonumber \\
		& =([J\Omega]_{\times}-J[\Omega]_{\times}-[\Omega]_{\times}J)\tilde{\Omega}_{o}-[J\tilde{\Omega}_{o}]_{\times}\tilde{\Omega}_{o}\nonumber \\
		& =S(\Omega)\tilde{\Omega}_{o}-[J\tilde{\Omega}_{o}]_{\times}\tilde{\Omega}_{o}\label{eq:VTOL_SOmega}
	\end{align}
	where $S(\Omega)=[J\Omega]_{\times}-J[\Omega]_{\times}-[\Omega]_{\times}J\in\mathfrak{so}\left(3\right)$,
	$[\Omega]_{\times}\tilde{\Omega}_{o}=-[\tilde{\Omega}_{o}]_{\times}\Omega$,
	and the identity in \eqref{eq:VTOL_R_Identity1} is employed. From
	\eqref{eq:VTOL_Translation}, \eqref{eq:VTOL_Perr}, and \eqref{eq:VTOL_Pestdot},
	one obtains
	\begin{align}
		\dot{\tilde{P}}_{o} & =\tilde{V}_{o}-k_{o2}\tilde{P}_{o}\label{eq:VTOL_Perr_dot}
	\end{align}
	From \eqref{eq:VTOL_Translation}, \eqref{eq:VTOL_Verr}, and \eqref{eq:VTOL_Vestdot},
	one has
	\begin{align}
		\dot{\tilde{V}}_{o} & =-k_{o3}\tilde{P}_{o}\label{eq:VTOL_Verr_dot}
	\end{align}
	Define the following cost function $L_{1}:\mathbb{SO}\left(3\right)\times\mathbb{R}^{3}\rightarrow\mathbb{R}_{+}$:
	\begin{equation}
		L_{1}=2||\tilde{R}_{o}||_{{\rm I}}+\frac{1}{2\gamma_{o}}\tilde{\Omega}_{o}^{\top}J\tilde{\Omega}_{o}\label{eq:VTOL_LyapC1}
	\end{equation}
	In view of \eqref{eq:VTOL_RerrI_dot}, \eqref{eq:VTOL_Omerr_dot},
	and the correction factors $w_{\Omega}$ and $w_{o}$ in \eqref{eq:VTOL_Wcorr},
	one finds that the derivative of \eqref{eq:VTOL_LyapC1} is as follows:
	\begin{align}
		\dot{L}_{1}= & -\mathbf{vex}(\boldsymbol{\mathcal{P}}_{a}(\tilde{R}_{o}))^{\top}(\tilde{\Omega}_{o}+Rw_{\Omega})\nonumber \\
		& +\frac{1}{\gamma_{o}}\tilde{\Omega}_{o}^{\top}(S(\Omega)\tilde{\Omega}_{o}-[J\tilde{\Omega}_{o}]_{\times}\tilde{\Omega}_{o}-\tilde{R}_{o}w_{o})\nonumber \\
		= & -k_{o1}||\mathbf{vex}(\boldsymbol{\mathcal{P}}_{a}(\tilde{R}_{o}))||^{2}\label{eq:VTOL_LyapC1dot}
	\end{align}
	where $[\tilde{\Omega}_{o}]_{\times}\tilde{\Omega}_{o}=0_{3\times1}$.
	It becomes clear that $\dot{L}_{1}$ is negative, continuous, and
	strictly decreasing, and consequently, $L_{1}$ is bounded indicating
	that $\mathbf{vex}(\boldsymbol{\mathcal{P}}_{a}(\tilde{R}_{o}))$
	and $\tilde{\Omega}_{o}$ are also bounded. Hence, $\ddot{L}_{1}$
	is bounded, and according to Barbalat Lemma, $\lim_{t\rightarrow\infty}\mathbf{vex}(\boldsymbol{\mathcal{P}}_{a}(\tilde{R}_{o}))=0_{3\times1}$
	shows that $||\tilde{R}_{o}||_{{\rm I}}\rightarrow0$, $||\dot{\tilde{R}}_{o}||_{{\rm I}}\rightarrow0$,
	$\lim_{t\rightarrow\infty}\dot{\tilde{R}}_{o}=0_{3\times3}$, $\lim_{t\rightarrow\infty}\tilde{R}_{o}=\mathbf{I}_{3}$,
	and $\lim_{t\rightarrow\infty}w_{\Omega}=\lim_{t\rightarrow\infty}w_{o}=0_{3\times1}$.
	Hence, $\lim_{t\rightarrow\infty}\tilde{\Omega}_{o}=0_{3\times1}$,
	and thereby, $\lim_{t\rightarrow\infty}L_{1}=0$. The derivative of
	the vex operator is equivalent to \cite{hashim2019AtiitudeSurvey}
	\begin{align}
		\mathbf{vex}(\boldsymbol{\mathcal{P}}_{a}(\dot{\tilde{R}}_{o})) & =-\frac{1}{2}\Psi(\tilde{R}_{o})(\tilde{\Omega}_{o}+Rw_{\Omega})\label{eq:VTOL_vex_dot}
	\end{align}
	where $\Psi(\tilde{R}_{o})={\rm Tr}\{\tilde{R}_{o}\}\mathbf{I}_{3}-\tilde{R}_{o}$.
	Recalling Lemma \ref{Lemm:vex_RI}, \eqref{eq:VTOL_vex_dot} and \eqref{eq:VTOL_Omerr_dot},
	one has
	\begin{align}
		& \frac{1}{2\delta_{o1}}\frac{d}{dt}\mathbf{vex}(\boldsymbol{\mathcal{P}}_{a}(\tilde{R}_{o}))^{\top}\tilde{\Omega}_{o}=-\frac{1}{4\delta_{o1}}(\tilde{\Omega}_{o}+Rw_{\Omega})^{\top}\Psi(\tilde{R}_{o})\tilde{\Omega}_{o}\nonumber \\
		& \hspace{1em}+\frac{1}{2\delta_{o1}}\mathbf{vex}(\boldsymbol{\mathcal{P}}_{a}(\tilde{R}_{o}))^{\top}J^{-1}(S(\Omega)\tilde{\Omega}_{o}-[J\tilde{\Omega}_{o}]_{\times}\tilde{\Omega}_{o}-\tilde{R}_{o}w_{o})\nonumber \\
		& \leq-\frac{2\gamma_{o}c_{o3}}{\delta_{o1}}||\tilde{R}_{o}||_{{\rm I}}-\frac{1}{4\delta_{o1}}||\tilde{\Omega}_{o}||^{2}+\frac{c_{o2}}{2\delta_{o1}}||\tilde{\Omega}_{o}||\sqrt{||\tilde{R}_{o}||_{{\rm I}}}\label{eq:VTOL_R_LyapC1_Aux}
	\end{align}
	where $\delta_{o1}$ is a positive constant, $\eta_{\Omega}=\sup_{t\geq0}S(\Omega)$,
	$c_{o1}=2\sqrt{1-||\tilde{R}_{o}(0)||_{{\rm I}}}$, $c_{o2}=\frac{c_{o1}(2\eta_{\Omega}+\underline{\lambda}_{J}\eta_{\Omega_{o}}+3\underline{\lambda}_{J}k_{o1})+0.5\eta_{\Omega_{o}}}{\underline{\lambda}_{J}}$,
	and $c_{o3}=\frac{c_{o1}^{2}}{\overline{\lambda}_{J}}$. In view of
	$L_{1}$ in \eqref{eq:VTOL_LyapC1}, define the following Lyapunov
	function candidate $\mathcal{L}_{o1}:\mathbb{SO}\left(3\right)\times\mathbb{R}^{3}\rightarrow\mathbb{R}_{+}$:
	\begin{equation}
		\mathcal{L}_{o1}=2||\tilde{R}_{o}||_{{\rm I}}+\frac{1}{2\gamma_{o}}\tilde{\Omega}_{o}^{\top}J\tilde{\Omega}_{o}+\frac{1}{2\delta_{o1}}\mathbf{vex}(\boldsymbol{\mathcal{P}}_{a}(\tilde{R}_{o}))^{\top}\tilde{\Omega}_{o}\label{eq:VTOL_Lyap_Lo1}
	\end{equation}
	Based on Lemma \ref{Lemm:vex_RI}, one finds{\small{}
		\[
		e_{o1}^{\top}\underbrace{\left[\begin{array}{cc}
				2 & -\frac{\overline{\lambda}_{J}c_{o1}}{4\delta_{o1}}\\
				-\frac{\overline{\lambda}_{J}c_{o1}}{4\delta_{o1}} & \frac{\underline{\lambda}_{J}}{2\gamma_{o}}
			\end{array}\right]}_{M_{1}}e_{o1}\leq\mathcal{L}_{o1}\leq e_{o1}^{\top}\underbrace{\left[\begin{array}{cc}
				2 & \frac{\overline{\lambda}_{J}c_{o1}}{4\delta_{o1}}\\
				\frac{\overline{\lambda}_{J}c_{o1}}{4\delta_{o1}} & \frac{\underline{\lambda}_{J}}{2\gamma_{o}}
			\end{array}\right]}_{M_{2}}e_{o1}
		\]
	}where $e_{o1}=[\sqrt{||\tilde{R}_{o}||_{{\rm I}}},||\tilde{\Omega}_{o}||]^{\top}$.
	The matrices $M_{1}$ and $M_{2}$ can be made positive by selecting
	$\delta_{o1}>\frac{\overline{\lambda}_{J}c_{o1}}{4}\sqrt{\frac{\gamma_{o}}{\underline{\lambda}_{J}}}$.
	Therefore, from \eqref{eq:VTOL_LyapC1}, \eqref{eq:VTOL_LyapC1dot},
	and \eqref{eq:VTOL_R_LyapC1_Aux}, the derivative of \eqref{eq:VTOL_Lyap_Lo1}
	becomes
	\begin{align}
		\dot{\mathcal{L}}_{o1} & \leq-\frac{1}{4\delta_{o1}}e_{o1}^{\top}\underbrace{\left[\begin{array}{cc}
				8(k_{o1}\delta_{o1}-\gamma_{o}c_{o3}) & c_{o2}\\
				c_{o2} & 1
			\end{array}\right]}_{A_{o1}}e_{o1}\label{eq:VTOL_Lyap_Lo1dot}
	\end{align}
	$A_{o1}$ can be made positive by selecting $\delta_{o1}>\frac{c_{o2}^{2}-\gamma_{o}c_{o3}}{8k_{o1}}$.
	By selecting $\delta_{o1}>\max\{\frac{\overline{\lambda}_{J}c_{o1}}{4}\sqrt{\frac{\gamma_{o}}{\underline{\lambda}_{J}}},\frac{c_{o2}^{2}+\gamma_{o}c_{o3}}{8k_{o1}}\}$
	and defining $\underline{\lambda}_{A_{o1}}$ as the minimum eigenvalue
	of $A_{o1}$, one has
	\begin{align}
		\dot{\mathcal{L}}_{o1} & \leq-\underline{\lambda}_{A_{o1}}||\tilde{R}||_{{\rm I}}-\underline{\lambda}_{A_{o1}}||\tilde{\Omega}_{o}||^{2}\label{eq:VTOL_Lyap_Lo1Final}
	\end{align}
	Consider the following Lyapunov function candidate $\mathcal{L}_{o2}:\mathbb{R}^{3}\times\mathbb{R}^{3}\rightarrow\mathbb{R}_{+}$:
	\begin{equation}
		\mathcal{L}_{o2}=\frac{1}{2}\tilde{P}_{o}^{\top}\tilde{P}_{o}+\frac{1}{2k_{o3}}\tilde{V}_{o}^{\top}\tilde{V}_{o}-\delta_{o2}\tilde{P}_{o}^{\top}\tilde{V}_{o}\label{eq:VTOL_Lyap_Lo2}
	\end{equation}
	It can be easily shown that $\mathcal{L}_{o2}$ follows
	\[
	e_{o2}^{\top}\underbrace{\left[\begin{array}{cc}
			\frac{1}{2} & -\frac{\delta_{o2}}{2}\\
			-\frac{\delta_{o2}}{2} & \frac{1}{2k_{o3}}
		\end{array}\right]}_{M_{3}}e_{o2}\leq\mathcal{L}_{o2}\leq e_{o2}^{\top}\underbrace{\left[\begin{array}{cc}
			\frac{1}{2} & \frac{\delta_{o2}}{2}\\
			\frac{\delta_{o2}}{2} & \frac{1}{2k_{o3}}
		\end{array}\right]}_{M_{4}}e_{o2}
	\]
	where $e_{o2}=[||\tilde{P}_{o}||,||\tilde{V}_{o}||]^{\top}$. It is
	evident that $M_{3}$ and $M_{4}$ are positive if $\delta_{o2}<\frac{1}{\sqrt{k_{o3}}}$.
	From \eqref{eq:VTOL_Perr_dot} and \eqref{eq:VTOL_Verr_dot}, one
	finds
	\begin{align}
		\dot{\mathcal{L}}_{o2}= & -k_{o2}\tilde{P}_{o}^{\top}\tilde{P}_{o}-\delta_{o2}\tilde{V}_{o}^{\top}\tilde{V}_{o}+\delta_{o2}k_{o2}\tilde{P}_{o}^{\top}\tilde{V}_{o}+\delta_{o2}k_{o3}\tilde{P}_{o}^{\top}\tilde{P}_{o}\nonumber \\
		\leq & -e_{o2}^{\top}\underbrace{\left[\begin{array}{cc}
				(k_{o2}-\delta_{o2}k_{o3}) & \frac{k_{o2}\delta_{o2}}{2}\\
				\frac{k_{o2}\delta_{o2}}{2} & \delta_{o2}
			\end{array}\right]}_{A_{o2}}e_{o2}\label{eq:VTOL_LyapL2_dot}
	\end{align}
	It becomes apparent that $A_{o2}$ is made positive by selecting $\delta_{o2}<\frac{4k_{o2}}{k_{o2}^{2}+k_{o3}}$.
	By selecting $\delta_{o2}<\min\{\frac{1}{\sqrt{k_{o3}}},\frac{4k_{o2}}{k_{o2}^{2}+k_{o3}}\}$
	and defining $\underline{\lambda}_{A_{o2}}$ as the minimum eigenvalue
	of $A_{o2}$, one shows that
	\begin{equation}
		\dot{\mathcal{L}}_{o2}\leq-\underline{\lambda}_{A_{o2}}||\tilde{P}_{o}||^{2}-\underline{\lambda}_{A_{o2}}||\tilde{V}_{o}||^{2}\label{eq:VTOL_Lyap_Lo2Final}
	\end{equation}
	From \eqref{eq:VTOL_Lyap_Lo1} and \eqref{eq:VTOL_Lyap_Lo2}, let
	us define the total Lyapunov function candidate for the observer design
	$\mathcal{L}_{oT}:\mathbb{SO}\left(3\right)\times\mathbb{R}^{3}\times\mathbb{R}^{3}\times\mathbb{R}^{3}\rightarrow\mathbb{R}_{+}$
	as follows:
	\begin{equation}
		\mathcal{L}_{oT}=\mathcal{L}_{o1}+\mathcal{L}_{o2}\label{eq:VTOL_LyapLo_Total}
	\end{equation}
	From \eqref{eq:VTOL_Lyap_Lo1Final} and \eqref{eq:VTOL_LyapL2_dot},
	one obtains
	\begin{equation}
		\dot{\mathcal{L}}_{oT}\leq-\underline{\lambda}_{A_{o}}(||\tilde{R}_{o}||_{{\rm I}}+||\tilde{\Omega}_{o}||^{2}+||\tilde{P}_{o}||^{2}+||\tilde{V}_{o}||^{2})\label{eq:VTOL_LyapLodot_Total}
	\end{equation}
	with $\underline{\lambda}_{A_{o}}=\min\{\underline{\lambda}_{A_{o1}},\underline{\lambda}_{A_{o2}}\}$.
	Define $\eta_{o}=\max\{\overline{\lambda}(M_{1}),\overline{\lambda}(M_{2}),\overline{\lambda}(M_{3}),\overline{\lambda}(M_{4})\}$.
	From \eqref{eq:VTOL_Lyap_Lo1}, \eqref{eq:VTOL_Lyap_Lo2}, \eqref{eq:VTOL_Lyap_Lo1Final},
	\eqref{eq:VTOL_Lyap_Lo2Final}, \eqref{eq:VTOL_LyapLo_Total}, and
	\eqref{eq:VTOL_LyapLodot_Total}, the following inequality is obtained:
	\begin{equation}
		\mathcal{L}_{oT}(t)\leq\mathcal{L}_{oT}(0)\exp(-\underline{\lambda}_{A_{o}}t/\eta_{o})\label{eq:VTOL_Lo_Total-2}
	\end{equation}
	such that $\lim_{t\rightarrow\infty}\tilde{R}_{o}=\mathbf{I}_{3}$,
	$\lim_{t\rightarrow\infty}\tilde{\Omega}_{o}=0_{3\times1}$, $\lim_{t\rightarrow\infty}\tilde{P}_{o}=0_{3\times1}$,
	and $\lim_{t\rightarrow\infty}\tilde{V}_{o}=0_{3\times1}$ exponentially.
	Consequently, the closed loop error signals of the observer design
	are uniformly almost globally exponentially stable and converge to
	the set $\mathcal{S}_{o}$ proving Theorem \ref{thm:Theorem1}.\end{proof}

\section{Observer-based Controller Scheme \label{sec:VTOL_Controller}}

As has been mentioned in the Introduction section, most of the existing
VTOL-UAV observer-based controllers utilize Euler angles representation
which is subject to singularity and fails to represent the attitude
at several configurations. Singularity of these methods leads to local
results. In addition, studies that use unit-quaternion suffer from
non-uniqueness in the attitude representation. This work, on the contrary,
designs the observer and the control laws using a Lie Group matrix
form which allows for unique and global attitude representation. The
objective of this Section is to design almost global control laws
for torque $\mathcal{T}\in\mathbb{R}^{3}$ and thrust $\Im\in\mathbb{R}$
to accurately track the VTOL-UAV position and velocity along the desired
trajectories using the estimates from Section \ref{sec:VTOL_Observer}:
$\hat{R}$, $\hat{\Omega}$, $\hat{P}$, and $\hat{V}$. Define $R_{d}\in\mathbb{SO}\left(3\right)$,
$\Omega_{d}\in\mathbb{R}^{3}$, $P_{d}\in\mathbb{R}^{3}$, and $V_{d}\in\mathbb{R}^{3}$
as the VTOL-UAV desired attitude, angular velocity, position, and
linear velocity, respectively. The proposed control strategy aims
to drive 
\begin{align*}
	R & \rightarrow R_{d}\\
	\Omega & \rightarrow\Omega_{d}\\
	P & \rightarrow P_{d}\\
	V & \rightarrow V_{d}
\end{align*}
Hence, define the errors in attitude, angular velocity, position,
and linear velocity as
\begin{align}
	\tilde{R}_{c}= & RR_{d}^{\top}\label{eq:VTOL_Rerr-c}\\
	\tilde{\Omega}_{c}= & \Omega-\tilde{R}_{c}\Omega_{d}\label{eq:VTOL_Omerr-c}\\
	\tilde{P}_{c}= & P-P_{d}\label{eq:VTOL_Perr-c}\\
	\tilde{V}_{c}= & V-V_{d}\label{eq:VTOL_Verr-c}
\end{align}
where $\tilde{R}_{c}\in\mathbb{SO}\left(3\right)$, $\tilde{\Omega}_{c}\in\mathbb{R}^{3}$,
$\tilde{P}_{c}\in\mathbb{R}^{3}$, and $\tilde{V}_{c}\in\mathbb{R}^{3}$.
The proposed control laws aim to achieve $\lim_{t\rightarrow\infty}\tilde{R}_{c}=\mathbf{I}_{3}$
and $\lim_{t\rightarrow\infty}\tilde{\Omega}_{c}=\lim_{t\rightarrow\infty}\tilde{P}_{c}=\lim_{t\rightarrow\infty}\tilde{V}_{c}=0_{3\times1}$.
Based on \eqref{eq:VTOL_Rotation}, the desired attitude dynamics
are as follows:
\begin{equation}
	\dot{R}_{d}=-[\Omega_{d}]_{\times}R_{d}\label{eq:VTOL_Rd_dot}
\end{equation}
Rewrite the velocity dynamics $\dot{V}=ge_{3}-\frac{\Im}{m}R^{\top}e_{3}$
in \eqref{eq:VTOL_Translation} as 
\begin{align*}
	\dot{V} & =ge_{3}-\frac{\Im}{m}R_{d}^{\top}e_{3}-\frac{\Im}{m}(R^{\top}-R_{d}^{\top})e_{3}\\
	& =F-\frac{\Im}{m}(R^{\top}-R_{d}^{\top})e_{3}
\end{align*}
where $F$ denotes an intermediary control input defined by
\begin{equation}
	F=ge_{3}-\frac{\Im}{m}R_{d}^{\top}e_{3}=[f_{1},f_{2},f_{3}]^{\top}\in\mathbb{R}^{3}\label{eq:VTOL_F_1}
\end{equation}
Consequently, $\Im=m||ge_{3}-F||$.
\begin{lem}
	\label{lem:Lemma_Qd}\cite{Roberts2009} Consider the dynamics in
	\eqref{eq:VTOL_Translation} and the intermediary control input in
	\eqref{eq:VTOL_F_1} with thrust magnitude $\Im=m||ge_{3}-F||$. The
	desired unit-quaternion components $Q_{d}=[q_{d0},q_{d}^{\top}]^{\top}\in\mathbb{S}^{3}$
	are as follows:
	\begin{equation}
		q_{d0}=\sqrt{\frac{m}{2\Im}(g-f_{3})+\frac{1}{2}},\hspace{1em}q_{d}=\left[\begin{array}{c}
			\frac{m}{2\Im q_{d0}}f_{2}\\
			-\frac{m}{2\Im q_{d0}}f_{1}\\
			0
		\end{array}\right]\label{eq:VTOL_Q_Qd}
	\end{equation}
	on condition that $F\neq[0,0,a]^{\top}$ for $a\geq g$. Note that
	$\mathbb{S}^{3}=\{Q_{d}\in\mathbb{R}^{4}|\,||Q_{d}||=1\}$ (see Subsection
	\ref{subsec:Unit-quaternion}). Let $F$ be differentiable with the
	desired angular velocity defined as
	\begin{align}
		\Omega_{d} & =\Xi(F)\dot{F}\label{eq:VTOL_Q_Omd}
	\end{align}
	and
	\begin{equation}
		\Xi(F)=\frac{1}{\alpha_{1}^{2}\alpha_{2}}\left[\begin{array}{ccc}
			-f_{1}f_{2} & -f_{2}^{2}+\alpha_{1}\alpha_{2} & f_{2}\alpha_{2}\\
			f_{1}^{2}-\alpha_{1}\alpha_{2} & f_{1}f_{2} & -f_{1}\alpha_{2}\\
			f_{2}\alpha_{1} & -f_{1}\alpha_{1} & 0
		\end{array}\right]\label{eq:VTOL_Q_Lambda}
	\end{equation}
	where $\alpha_{1}=||ge_{3}-F||$ and $\alpha_{2}=||ge_{3}-F||+g-f_{3}$. \hspace{10pt} $\blacksquare$
\end{lem}
From Lemma \ref{lem:Lemma_Qd} and given the desired unit-quaternion
$Q_{d}=[q_{d0},q_{d}^{\top}]^{\top}\in\mathbb{S}^{3}$, the desired
attitude $R_{d}=\mathcal{R}_{Q_{d}}$ can be obtained from the map
in \eqref{eq:NAV_Append_SO3} as follows:
\[
\mathcal{R}_{Q_{d}}=(q_{d0}^{2}-||q_{d}||^{2})\mathbf{I}_{3}+2q_{d}q_{d}^{\top}-2q_{d0}[q_{d}]_{\times}\in\mathbb{SO}\left(3\right)
\]
It is worth noting that Lemma \ref{lem:Lemma_Qd} indicates that $Q_{d}$
and $\Im$ are singularity-free.
\begin{rem}
	In this work $F$ is designed to be twice differentiable such that
	the desired angular velocity rate of change $\dot{\Omega}_{d}$ can
	be defined as
	\begin{equation}
		\dot{\Omega}_{d}=\dot{\Xi}(F)\dot{F}+\Xi(F)\ddot{F}\label{eq:VTOL_Q_Omd_dot}
	\end{equation}
	$\dot{F}$ and $\ddot{F}$ are provided in the Appendix.
\end{rem}
Let us define the following variables:
\begin{equation}
	\mathcal{E}=\tilde{P}_{c}-\theta,\hspace{2em}\dot{\mathcal{E}}=\tilde{V}_{c}-\dot{\theta}\label{eq:VTOL_Q_E}
\end{equation}
where $\theta\in\mathbb{R}^{3}$ stands for an auxiliary variable.
Consider proposing the following VTOL-UAV control strategy:
\begin{align}
	\mathcal{T}= & k_{c1}\mathbf{vex}(\boldsymbol{\mathcal{P}}_{a}(\tilde{R}_{c}))-k_{c2}(\tilde{R}_{o}\hat{\Omega}-\tilde{R}_{c}\Omega_{d})+J\tilde{R}_{c}\dot{\Omega}_{d}\nonumber \\
	& +[\tilde{R}_{c}\Omega_{d}]_{\times}J\tilde{R}_{c}\Omega_{d}\label{eq:VTOL_Tau}\\
	\ddot{\theta}= & -k_{\theta1}\psi(\theta)-k_{\theta2}\psi(\dot{\theta})+k_{c3}(\hat{P}-P_{d}-\theta)\nonumber \\
	& +k_{c4}(\hat{V}-V_{d}-\dot{\theta})\label{eq:VTOL_Theta_dot}\\
	F= & \ddot{P}_{d}-k_{\theta1}\psi(\theta)-k_{\theta2}\psi(\dot{\theta})\label{eq:VTOL_F}\\
	\Im= & m||ge_{3}-F||\label{eq:VTOL_Thrust}
\end{align}
where $\mathcal{T}$ denotes the torque input, $J$ denotes the inertia
matrix of the vehicle, $\tilde{R}_{o}=R_{y}\hat{R}^{\top}$, $\tilde{R}_{c}=R_{y}R_{d}^{\top}$,
$R_{y}$ represents the reconstructed attitude, $\theta\in\mathbb{R}^{3}$
stands for an auxiliary variable, $F$ denotes the intermediary control
input (selected as in \cite{abdessameud2013position}), $g$ and $m$
are defined in \eqref{eq:VTOL_Translation}, $\Im\in\mathbb{R}$ is
the magnitude of thrust, $\hat{\Omega}$, $\hat{P}$, and $\hat{V}$
denote the estimates of angular velocity, position, and linear velocity,
respectively, $\ddot{P}_{d}$ denotes the second derivative of the
desired position, $\psi(\theta)$ and $\psi(\dot{\theta})$ are the
bounded functions defined in the Appendix, and $k_{\theta1}$, $k_{\theta2}$,
$k_{c1}$, $k_{c2}$, $k_{c3}$, and $k_{c4}$ are strictly positive
constants. 
\begin{thm}
	\label{thm:Theorem2}Consider combining the true VTOL-UAV nonlinear
	dynamics in \eqref{eq:VTOL_Rotation} and \eqref{eq:VTOL_Translation}
	and the nonlinear observer in \eqref{eq:VTOL_ObsvCompact} with the
	control laws in \eqref{eq:VTOL_Tau} and \eqref{eq:VTOL_Thrust}.
	Let Assumption \ref{Assum:P_Om_desired} hold and let $\tilde{R}_{c}(0)\notin\mathcal{U}_{s}$.
	Define the following set:
	\begin{align}
		\mathcal{S}_{c}= & \{(\tilde{R}_{c},\tilde{\Omega}_{c},\tilde{P}_{c},\tilde{V}_{c})\in\mathbb{SO}\left(3\right)\times\mathbb{R}^{3}\times\mathbb{R}^{3}\times\mathbb{R}^{3}\,|\nonumber \\
		& \hspace{6em}\tilde{R}_{c}=\mathbf{I}_{3},\tilde{\Omega}_{c}=\tilde{P}_{c}=\tilde{V}_{c}=0_{3\times1}\}\label{eq:VTOL_Set2}
	\end{align}
	Then, the set $\mathcal{S}_{c}$ is uniformly almost globally exponentially
	stable.
\end{thm}
\begin{proof}The statement in Theorem \ref{thm:Theorem2} entails
	that the closed loop error signals converge exponentially to the equilibrium
	point starting from any initial condition except for the three repeller
	attitude cases given in Definition \ref{def:Unstable-set} ($\tilde{R}_{c}(0)\notin\mathcal{U}_{s}$).
	Consider the attitude error $\tilde{R}_{c}=RR_{d}^{\top}$ in \eqref{eq:VTOL_Rerr-c}.
	Using \eqref{eq:VTOL_Rotation} and \eqref{eq:VTOL_Rd_dot}, one shows
	that
	\begin{align}
		\dot{\tilde{R}}_{c} & =R\dot{R}_{d}^{\top}+\dot{R}R_{d}^{\top}=\tilde{R}_{c}[\Omega_{d}]_{\times}-[\Omega]_{\times}\tilde{R}_{c}\nonumber \\
		& =-[\Omega-\tilde{R}_{c}\Omega_{d}]_{\times}\tilde{R}_{c}=-[\tilde{\Omega}_{c}]_{\times}\tilde{R}_{c}\label{eq:VTOL_Rcerr_dot}
	\end{align}
	where the identity in \eqref{eq:VTOL_R_Identity1} was used, and $\tilde{\Omega}_{c}$
	is defined in \eqref{eq:VTOL_Omerr-c}. In view of \eqref{eq:VTOL_Ecul_Dist},
	\eqref{eq:VTOL_R_Identity2}, and \eqref{eq:VTOL_Rcerr_dot}, one
	finds \cite{hashim2019SO3Wiley}
	\begin{align}
		||\dot{\tilde{R}}_{c}||_{{\rm I}}= & -\frac{1}{4}{\rm Tr}\{\dot{\tilde{R}}_{c}\}=\frac{1}{4}{\rm Tr}\{[\tilde{\Omega}_{c}]_{\times}\boldsymbol{\mathcal{P}}_{a}(\tilde{R}_{c})\}\nonumber \\
		= & -\frac{1}{2}\mathbf{vex}(\boldsymbol{\mathcal{P}}_{a}(\tilde{R}_{c}))^{\top}\tilde{\Omega}_{c}\label{eq:VTOL_RerrIc_dot}
	\end{align}
	In view of \eqref{eq:VTOL_Rotation}, \eqref{eq:VTOL_Omerr-c}, and
	\eqref{eq:VTOL_Rcerr_dot}, one has
	\begin{align}
		J\dot{\tilde{\Omega}}_{c}= & J\dot{\Omega}-J\dot{\tilde{R}}_{c}\Omega_{d}-J\tilde{R}_{c}\dot{\Omega}_{d}\nonumber \\
		= & [J\Omega]_{\times}\Omega+\mathcal{T}+J[\tilde{\Omega}_{c}]_{\times}\tilde{R}_{c}\Omega_{d}-J\tilde{R}_{c}\dot{\Omega}_{d}\nonumber \\
		= & S(\Omega)\tilde{\Omega}_{c}-[J\tilde{\Omega}_{c}]_{\times}\tilde{\Omega}_{c}+\mathcal{T}\nonumber \\
		& +[J\tilde{R}_{c}\Omega_{d}]_{\times}\tilde{R}_{c}\Omega_{d}-J\tilde{R}_{c}\dot{\Omega}_{d}\label{eq:VTOL_OmerrC_dot}
	\end{align}
	with
	\begin{align}
		& \left[J\Omega\right]_{\times}\Omega+J[\tilde{\Omega}_{c}]_{\times}\tilde{R}_{c}\Omega_{d}-[J\tilde{R}_{c}\Omega_{d}]_{\times}\tilde{R}_{c}\Omega_{d}\nonumber \\
		& =\left[J\Omega\right]_{\times}\Omega+(J[\tilde{\Omega}_{c}]_{\times}-[J\Omega]_{\times}+[J\tilde{\Omega}_{c}]_{\times})\tilde{R}_{c}\Omega_{d}\nonumber \\
		& =\left[J\Omega\right]_{\times}\tilde{\Omega}_{c}-J\left[\Omega\right]_{\times}\tilde{\Omega}_{c}-\left[\Omega\right]_{\times}J\tilde{\Omega}_{c}-[J\tilde{\Omega}_{c}]_{\times}\tilde{\Omega}_{c}\nonumber \\
		& =S(\Omega)\tilde{\Omega}_{c}-[J\tilde{\Omega}_{c}]_{\times}\tilde{\Omega}_{c}\label{eq:VTOL_SOmegaC}
	\end{align}
	using the fact that $[\Omega]_{\times}\tilde{\Omega}_{c}=-[\tilde{\Omega}_{c}]_{\times}\Omega$
	and $S(\Omega)=\left[J\Omega\right]_{\times}-J\left[\Omega\right]_{\times}-\left[\Omega\right]_{\times}J\in\mathfrak{so}\left(3\right)$.
	From \eqref{eq:VTOL_Perr-c} and \eqref{eq:VTOL_Translation}, one
	has
	\begin{align}
		\dot{\tilde{P}}_{c} & =\tilde{V}_{c}\label{eq:VTOL_PerrC_dot}
	\end{align}
	where $\dot{P}_{d}=V_{d}$. In the same spirit, from \eqref{eq:VTOL_Verr-c},
	\eqref{eq:VTOL_Translation}, and \eqref{eq:VTOL_F_1}, one shows
	that
	\begin{align}
		\dot{\tilde{V}}_{c} & =F-||ge_{3}-F||(R^{\top}-R_{d}^{\top})e_{3}-\ddot{P}_{d}\label{eq:VTOL_VerrC_dot}
	\end{align}
	Define the real-valued function $L_{2}:\mathbb{SO}\left(3\right)\times\mathbb{R}^{3}\rightarrow\mathbb{R}_{+}$
	\begin{equation}
		L_{2}=2||\tilde{R}_{c}||_{{\rm I}}+\frac{1}{k_{c1}}\tilde{\Omega}_{c}^{\top}J\tilde{\Omega}_{c}\label{eq:VTOL_LyapC2}
	\end{equation}
	Considering \eqref{eq:VTOL_Rerr-c}, \eqref{eq:VTOL_Omerr-c}, and
	\eqref{eq:VTOL_Tau}, the derivative of \eqref{eq:VTOL_LyapC2} is
	\begin{align}
		\dot{L}_{2}= & -\frac{1}{2}\mathbf{vex}(\boldsymbol{\mathcal{P}}_{a}(\tilde{R}_{c}))^{\top}\tilde{\Omega}_{c}\nonumber \\
		& +\frac{1}{k_{c3}}\tilde{\Omega}_{c}^{\top}(\mathcal{T}+[J\tilde{R}_{c}\Omega_{d}]_{\times}\tilde{R}_{c}\Omega_{d}-J\tilde{R}_{c}\dot{\Omega}_{d})\nonumber \\
		\leq & -\frac{k_{c2}}{k_{c1}}||\tilde{\Omega}_{c}||^{2}+\frac{k_{c2}}{k_{c1}}||\tilde{\Omega}_{c}||\,||\tilde{\Omega}_{o}||\label{eq:VTOL_LyapC2_dot}
	\end{align}
	with $[\tilde{\Omega}_{c}]_{\times}\tilde{\Omega}_{c}=0_{3\times1}$.
	It follows from \eqref{eq:VTOL_LyapLodot_Total} that $\tilde{\Omega}_{o}$
	is bounded and converges to zero. Hence, $\tilde{\Omega}_{c}$ is
	bounded and $\lim_{t\rightarrow\infty}||\tilde{\Omega}_{c}||=0$.
	Given that $\tilde{\Omega}_{o}$, $\tilde{\Omega}_{c}$, $\Omega_{d}$,
	$\dot{\Omega}_{d}$, and $\dot{L}_{2}$ are bounded, consider the
	following derivative of the vex operator \cite{hashim2019AtiitudeSurvey}:
	\begin{align}
		\mathbf{vex}(\boldsymbol{\mathcal{P}}_{a}(\dot{\tilde{R}}_{c})) & =-\frac{1}{2}\Psi(\tilde{R}_{c})\tilde{\Omega}_{c}\label{eq:VTOL_vex_dot-1}
	\end{align}
	where $\Psi(\tilde{R}_{c})={\rm Tr}\{\tilde{R}_{c}\}\mathbf{I}_{3}-\tilde{R}_{c}$.
	Now, let us find the derivative
	\begin{align}
		& -\frac{1}{2\delta_{c1}}\frac{d}{dt}\mathbf{vex}(\boldsymbol{\mathcal{P}}_{a}(\tilde{R}_{c}))^{\top}\tilde{\Omega}_{c}\nonumber \\
		& \leq-\frac{k_{c1}c_{c2}}{2\delta_{c1}}||\tilde{R}_{c}||_{{\rm I}}+\frac{c_{c3}||\tilde{\Omega}_{c}||+c_{c4}||\tilde{\Omega}_{o}||}{2\delta_{c1}}\sqrt{||\tilde{R}_{c}||_{{\rm I}}}\label{eq:VTOL_AUX2}
	\end{align}
	where $\eta_{\Omega_{c}}=\sup_{t\geq0}||J\tilde{\Omega}_{c}||$, $\eta_{\Omega}=\sup_{t\geq0}S(\Omega)$,
	$c_{c1}=\sqrt{1-||\tilde{R}_{c}(0)||_{{\rm I}}}$, $c_{c2}=\frac{c_{c1}^{2}}{\overline{\lambda}_{J}}$,
	$c_{c3}=\frac{(\eta_{\Omega}+k_{c2})c_{c1}+(1+\underline{\lambda}_{J})\eta_{\Omega_{c}}}{\underline{\lambda}_{J}}$,
	and $c_{c4}=\frac{k_{c2}c_{c1}}{\underline{\lambda}_{J}}$. Based
	on \eqref{eq:VTOL_LyapC2} and \eqref{eq:VTOL_AUX2}, consider the
	following Lyapunov function candidate $\mathcal{L}_{c1}:\mathbb{SO}\left(3\right)\times\mathbb{R}^{3}\rightarrow\mathbb{R}_{+}$:
	\begin{equation}
		\mathcal{L}_{c1}=2||\tilde{R}_{c}||_{{\rm I}}+\frac{1}{2k_{c1}}\tilde{\Omega}_{c}^{\top}J\tilde{\Omega}_{c}-\frac{1}{2\delta_{c1}}\mathbf{vex}(\boldsymbol{\mathcal{P}}_{a}(\tilde{R}_{c}))^{\top}\tilde{\Omega}_{c}\label{eq:VTOL_LyapLc2}
	\end{equation}
	such that{\small
		\[
		e_{c1}^{\top}\underbrace{\left[\begin{array}{cc}
				2 & -\frac{\overline{\lambda}_{J}c_{c1}}{4\delta_{c1}}\\
				-\frac{\overline{\lambda}_{J}c_{c1}}{4\delta_{c1}} & \frac{\underline{\lambda}_{J}}{2k_{c1}}
			\end{array}\right]}_{M_{5}}e_{c1}\leq\mathcal{L}_{c1}\leq e_{c1}^{\top}\underbrace{\left[\begin{array}{cc}
				2 & \frac{\overline{\lambda}_{J}c_{c1}}{4\delta_{c1}}\\
				\frac{\overline{\lambda}_{J}c_{c1}}{4\delta_{c1}} & \frac{\underline{\lambda}_{J}}{2k_{c1}}
			\end{array}\right]}_{M_{6}}e_{c1}
		\]
	}where $e_{c1}=[\sqrt{||\tilde{R}_{c}||_{{\rm I}}},||\tilde{\Omega}_{c}||]^{\top}$.
	$M_{5}$ and $M_{6}$ are positive if $\delta_{c1}>\frac{\overline{\lambda}_{J}c_{c1}}{4}\sqrt{\frac{k_{c1}}{\underline{\lambda}_{J}}}$.
	Based on \eqref{eq:VTOL_LyapLc2}, \eqref{eq:VTOL_LyapC2_dot}, and
	\eqref{eq:VTOL_AUX2}, one finds
	\begin{align}
		\dot{\mathcal{L}}_{c1}\leq & -e_{c1}^{\top}\underbrace{\left[\begin{array}{cc}
				\frac{k_{c1}c_{c2}}{2\delta_{c1}} & \frac{c_{c3}}{4\delta_{c1}}\\
				\frac{c_{c3}}{4\delta_{c1}} & \frac{k_{c2}}{k_{c1}}
			\end{array}\right]}_{A_{c1}}e_{c1}+\frac{c_{c4}}{2\delta_{c1}}||\tilde{\Omega}_{o}||\sqrt{||\tilde{R}_{c}||_{{\rm I}}}\nonumber \\
		& +\frac{k_{c2}}{k_{c1}}||\tilde{\Omega}_{o}||\,||\tilde{\Omega}_{c}||\label{eq:VTOL_LyapLc2dot}
	\end{align}
	$A_{c1}$ is positive if $\delta_{c1}>\frac{c_{c3}^{2}}{8k_{c2}c_{c2}}$.
	Let us set $\delta_{c1}>\max\{\frac{\overline{\lambda}_{J}c_{c1}}{4}\sqrt{\frac{k_{c1}}{\underline{\lambda}_{J}}},\frac{c_{c3}^{2}}{8k_{c2}c_{c2}}\}$
	with $\underline{\lambda}_{A_{c1}}$ being the minimum eigenvalue
	of $A_{c1}$. Recalling \eqref{eq:VTOL_Q_E} and \eqref{eq:VTOL_VerrC_dot},
	one shows
	\begin{equation}
		\begin{cases}
			\dot{\mathcal{E}} & =\tilde{V}_{c}-\dot{\theta}\\
			\ddot{\mathcal{E}} & =F-||ge_{3}-F||(R^{\top}-R_{d}^{\top})e_{3}-\ddot{P}_{d}-\ddot{\theta}
		\end{cases}\label{eq:VTOL_Edot}
	\end{equation}
	In view of \eqref{eq:VTOL_Theta_dot} and \eqref{eq:VTOL_F}, it becomes
	apparent that $\ddot{\theta}$ and $F$ are bounded indicating that
	$\Im$ is bounded. Also, note that $||\mathbf{I}_{3}-\tilde{R}_{c}||_{F}=2\sqrt{2}\sqrt{||\tilde{R}_{c}||_{{\rm I}}}$
	such that $||ge_{3}-F||(R^{\top}-R_{d}^{\top})e_{3}=||ge_{3}-F||(\mathbf{I}_{3}-\tilde{R}_{c})R^{\top}e_{3}\leq4(||ge_{3}||+\sup_{t\geq0}\{||\ddot{P}_{d}||\}+(k_{\theta1}+k_{\theta2})\sqrt{||\tilde{R}_{c}||_{{\rm I}}}\triangleq4\pi\,\sqrt{||\tilde{R}_{c}||_{{\rm I}}}$
	\begin{align}
		||ge_{3}-F||(R^{\top}-R_{d}^{\top})e_{3}\leq & 4\pi\,\sqrt{||\tilde{R}_{c}||_{{\rm I}}}\label{eq:VTOL_RHO}
	\end{align}
	where $\pi$ is an upper bounded positive constant. From \eqref{eq:VTOL_Q_E}
	and \eqref{eq:VTOL_Edot}, define the following Lyapunov function
	candidate:
	\begin{equation}
		\mathcal{L}_{c2}=\frac{1}{2}\mathcal{E}^{\top}\mathcal{E}+\frac{1}{2k_{c1}}\dot{\mathcal{E}}^{\top}\dot{\mathcal{E}}+\frac{1}{\delta_{c2}}\mathcal{E}^{\top}\dot{\mathcal{E}}\label{eq:VTOL_LyapL1C}
	\end{equation}
	such that
		\[
		e_{c2}^{\top}\underbrace{\left[\begin{array}{cc}
				\frac{1}{2} & \frac{-1}{2\delta_{c2}}\\
				\frac{-1}{2\delta_{c2}} & \frac{1}{2k_{c3}}
			\end{array}\right]}_{M_{7}}e_{c2}\leq\mathcal{L}_{c2}\leq e_{c2}^{\top}\underbrace{\left[\begin{array}{cc}
				\frac{1}{2} & \frac{1}{2\delta_{c2}}\\
				\frac{1}{2\delta_{c2}} & \frac{1}{2k_{c3}}
			\end{array}\right]}_{M_{8}}e_{c2}
		\]
	where $e_{c2}=[||\mathcal{E}||,||\dot{\mathcal{E}}||]^{\top}$. $M_{7}$
	and $M_{8}$ are made positive by selecting $\delta_{c2}>\sqrt{k_{c3}}$.
	Using \eqref{eq:VTOL_Edot}, \eqref{eq:VTOL_F}, \eqref{eq:VTOL_Theta_dot},
	and \eqref{eq:VTOL_RHO}, one obtains
	\begin{align}
		& \dot{\mathcal{L}}_{c2}\leq-e_{c2}^{\top}\underbrace{\left[\begin{array}{cc}
				\frac{k_{c3}}{\delta_{c2}} & \frac{k_{c4}}{2\delta_{c2}}\\
				\frac{k_{c4}}{2\delta_{c2}} & \frac{k_{c4}}{k_{c3}}-\frac{1}{\delta_{c2}}
			\end{array}\right]}_{A_{c2}}e_{c2}+(||\dot{\mathcal{E}}||+\frac{k_{c3}}{\delta_{c2}}||\mathcal{E}||)||\tilde{P}_{o}||\nonumber \\
		& +(\frac{k_{c4}}{k_{c3}}||\dot{\mathcal{E}}||+\frac{k_{c4}}{\delta_{c2}}||\mathcal{E}||)||\tilde{V}_{o}||+4\pi(||\dot{\mathcal{E}}||+\frac{1}{\delta_{c2}}||\mathcal{E}||)\sqrt{||\tilde{R}_{c}||_{{\rm I}}}\label{eq:VTOL_LyapL1Cdot}
	\end{align}
	Note that $\hat{P}-P_{d}-\theta=-\tilde{P}_{o}+\mathcal{E}$ and $\hat{V}-V_{d}-\dot{\theta}=-\tilde{V}_{o}+\dot{\mathcal{E}}$.
	$A_{c2}$ is made positive by selecting $\delta_{c2}>\frac{k_{c4}^{2}+4k_{c3}}{4k_{c4}}$.
	Consider selecting $\delta_{c2}>\max\{\sqrt{k_{c3}},\frac{k_{c4}^{2}+4k_{c3}}{4k_{c4}}\}$.
	Based on \eqref{eq:VTOL_LyapLodot_Total}, $\tilde{P}_{o}$ and $\tilde{V}_{o}$
	are bounded and converge to zero, while from \eqref{eq:VTOL_LyapC2_dot},
	$\sqrt{||\tilde{R}_{c}||_{{\rm I}}}$ is bounded. Therefore, $\mathcal{L}_{c2}$
	is bounded. Recall \eqref{eq:VTOL_LyapL1C}, \eqref{eq:VTOL_LyapLc2},
	and define the following Lyapunov function candidate $\mathcal{L}_{cT}:\mathbb{SO}\left(3\right)\times\mathbb{R}^{3}\times\mathbb{R}^{3}\times\mathbb{R}^{3}\rightarrow\mathbb{R}_{+}$:
	\begin{equation}
		\mathcal{L}_{cT}=\mathcal{L}_{c1}+\mathcal{L}_{c2}\label{eq:VTOL_Q_LyapLc-Final}
	\end{equation}
	From \eqref{eq:VTOL_LyapL1Cdot} and \eqref{eq:VTOL_LyapLc2dot},
	one obtains
	\begin{align}
		\dot{\mathcal{L}}_{cT}\leq & -e_{c3}^{\top}\underbrace{\left[\begin{array}{cc}
				\underline{\lambda}_{A_{c1}} & -c_{e}\\
				-c_{e} & \underline{\lambda}_{A_{c2}}
			\end{array}\right]}_{A_{c}}e_{c3}+(||\dot{\mathcal{E}}||+\frac{k_{c3}}{\delta_{c2}}||\mathcal{E}||)||\tilde{P}_{o}||\nonumber \\
		& +(\frac{k_{c2}}{k_{c1}}||\tilde{\Omega}_{c}||+\frac{c_{c4}}{2\delta_{c1}}\sqrt{||\tilde{R}_{c}||_{{\rm I}}})||\tilde{\Omega}_{o}||\nonumber \\
		& +(\frac{k_{c4}}{k_{c3}}||\dot{\mathcal{E}}||+\frac{k_{c4}}{\delta_{c2}}||\mathcal{E}||)||\tilde{V}_{o}||\label{eq:VTOL_Q_LyapLcdot-Final}
	\end{align}
	where $c_{e}=\max\{2\pi,\frac{2\pi}{\delta_{c2}}\}$ and $e_{c3}=[||e_{c1}||,||e_{c2}||]^{\top}$.
	For a positive definite $A_{c}$, select $\underline{\lambda}_{A_{c1}}>c_{e}^{2}/\underline{\lambda}_{A_{c2}}$.
	Let $\underline{\lambda}_{A_{c}}$ denote the minimum eigenvalue of
	$A_{c}$. Using \eqref{eq:VTOL_LyapLo_Total} and \eqref{eq:VTOL_Q_LyapLc-Final},
	define the following Lyapunov function candidate: 
	\begin{equation}
		\mathcal{L}_{T}=\mathcal{L}_{oT}+\mathcal{L}_{cT}\label{eq:VTOL_Q_LyapL-Final}
	\end{equation}
	Thus, from \eqref{eq:VTOL_LyapLodot_Total} and \eqref{eq:VTOL_Q_LyapLcdot-Final},
	one finds
	\begin{align}
		\dot{\mathcal{L}}_{T}\leq & -\underline{\lambda}_{A_{o}}||\tilde{R}_{o}||_{{\rm I}}-e_{1}^{\top}\underbrace{\left[\begin{array}{ccc}
				\underline{\lambda}_{A_{o}} & \frac{c_{c4}}{4\delta_{c1}} & \frac{k_{c2}}{2k_{c1}}\\
				\frac{c_{c4}}{4\delta_{c1}} & \underline{\lambda}_{A_{c}} & 0\\
				\frac{k_{c2}}{2k_{c1}} & 0 & \underline{\lambda}_{A_{c}}
			\end{array}\right]}_{A_{1}}e_{1}\nonumber \\
		& -e_{2}^{\top}\underbrace{\left[\begin{array}{cc}
				\underline{\lambda}_{A_{o}}\mathbf{I}_{2} & \frac{c_{p}}{2}\mathbf{I}_{2}\\
				\frac{c_{p}}{2}\mathbf{I}_{2} & \underline{\lambda}_{A_{c}}\mathbf{I}_{2}
			\end{array}\right]}_{A_{2}}e_{2}\label{eq:VTOL_Q_LyapL-Final-dot}
	\end{align}
	where $e_{1}=[||\tilde{\Omega}_{o}||,\sqrt{||\tilde{R}_{c}||_{{\rm I}}},||\tilde{\Omega}_{c}||]^{\top}$,
	$e_{2}=[||\tilde{P}_{o}||,||\tilde{V}_{o}||,||\mathcal{E}||,||\dot{\mathcal{E}}||]^{\top}$,
	and $c_{p}=\max\{1,\frac{k_{c3}}{\delta_{c2}},\frac{k_{c4}}{k_{c3}},\frac{k_{c4}}{\delta_{c2}}\}$.
	$A_{1}$ is positive if $\underline{\lambda}_{A_{o}}>\frac{4\delta_{c1}^{2}k_{c2}^{2}+k_{c1}^{2}c_{c4}^{2}}{16\underline{\lambda}_{A_{c}}k_{c1}^{2}\delta_{c1}^{2}}$,
	and $A_{2}$ is positive if $\underline{\lambda}_{A_{o}}>\frac{c_{p}^{2}}{4\underline{\lambda}_{A_{c}}}$.
	Thus, consider selecting $\underline{\lambda}_{A_{o}}>\max\{\frac{c_{p}^{2}}{4\underline{\lambda}_{A_{c}}},\frac{4\delta_{c1}^{2}k_{c2}^{2}+k_{c1}^{2}c_{c4}^{2}}{16\underline{\lambda}_{A_{c}}k_{c1}^{2}\delta_{c1}^{2}}\}$.
	Let $\underline{\lambda}_{A_{1}}$ and $\underline{\lambda}_{A_{2}}$
	denote the minimum eigenvalue of $A_{1}$ and $A_{2}$, respectively.
	By defining $\underline{\lambda}_{A}=\min\{\underline{\lambda}_{A_{1}},\underline{\lambda}_{A_{2}},\underline{\lambda}_{A_{o}}\}$
	and $\overline{\lambda}_{M}=\max\{\overline{\lambda}(M_{1}),\overline{\lambda}(M_{2}),\ldots,\overline{\lambda}(M_{8})\}$,
	one finds
	\begin{align}
		\dot{\mathcal{L}}_{T}\leq & -(\underline{\lambda}_{A}/\overline{\lambda}_{M})\mathcal{L}_{T}\nonumber \\
		\mathcal{L}_{T}(t)\leq & \mathcal{L}_{T}(0)\exp(-t\underline{\lambda}_{A}/\overline{\lambda}_{M}),\hspace{1em}\forall t\geq0\label{eq:VTOL_Q_LyapL-dot-Total}
	\end{align}
	such that $\lim_{t\rightarrow\infty}\tilde{R}_{o}=\lim_{t\rightarrow\infty}\tilde{R}_{c}=\mathbf{I}_{3}$,
	$\lim_{t\rightarrow\infty}||\tilde{\Omega}_{o}||=\lim_{t\rightarrow\infty}||\tilde{P}_{o}||=\lim_{t\rightarrow\infty}||\tilde{V}_{o}||=0$,
	and $\lim_{t\rightarrow\infty}||\tilde{\Omega}_{c}||=\lim_{t\rightarrow\infty}||\mathcal{E}||=\lim_{t\rightarrow\infty}||\dot{\mathcal{E}}||=0$.
	From \eqref{eq:VTOL_Q_LyapL-dot-Total}, the definition of $\ddot{\theta}$
	in \eqref{eq:VTOL_Theta_dot} implies that $\ddot{\theta}\rightarrow-k_{\theta1}\psi(\theta)-k_{\theta2}\psi(\dot{\theta})$
	as $\mathcal{E},\dot{\mathcal{E}}\rightarrow0$ and, in turn, $||\psi(\theta)||$
	and $||\psi(\dot{\theta})||$ become strictly decreasing with $\psi(\theta),\psi(\dot{\theta})\rightarrow0$
	which shows that $\lim_{t\rightarrow\infty}\theta=\lim_{t\rightarrow\infty}\dot{\theta}=0$.
	Therefore, $\lim_{t\rightarrow\infty}\tilde{R}_{o}=\lim_{t\rightarrow\infty}\tilde{R}_{c}=\mathbf{I}_{3}$,
	$\lim_{t\rightarrow\infty}||\tilde{\Omega}_{o}||=\lim_{t\rightarrow\infty}||\tilde{P}_{o}||=\lim_{t\rightarrow\infty}||\tilde{V}_{o}||=0$,
	$\lim_{t\rightarrow\infty}||\tilde{\Omega}_{c}||=\lim_{t\rightarrow\infty}||\tilde{P}_{c}||=\lim_{t\rightarrow\infty}||\tilde{V}_{c}||=0$,
	and the closed loop error signals of the observer-based controller
	design are uniformly almost globally exponentially stable proving
	Theorem \ref{thm:Theorem2}.\end{proof}

\section{Implementation Steps \label{sec:VTOL_Implementation}}

The VTOL-UAV observer-based controller on the Lie Group is presented
in a discrete form to facilitate the implementation process. Define
$\Delta t$ as a small sample time step. The implementation steps
are as follows:

\textcolor{blue}{\textbf{Step 1.}} Select $\hat{\Omega}_{1},\hat{P}_{0},\hat{V}_{0},\theta_{0},\dot{\theta}_{0}\in\mathbb{R}^{3}$,
$\hat{R}_{0}\in\mathbb{SO}(3)$, formulate the navigation matrix $\hat{X}_{0}=\left[\begin{array}{ccc}
	\hat{R}_{0}^{\top} & \hat{P}_{0} & \hat{V}_{0}\\
	0_{1\times3} & 1 & 0\\
	0_{1\times3} & 0 & 1
\end{array}\right]$, and define $k=1$.\vspace{0.2cm}

\textcolor{blue}{\textbf{Step 2.}} (Pose reconstruction) Use one of the methods of pose
reconstruction to obtain reconstructed attitude $R_{y|k}$ and position
$P_{y|k}$. For more details consult \cite{hashim2019SO3Wiley, hashim2020SE3Stochastic}.\vspace{0.2cm}

\textcolor{blue}{\textbf{Step 3.}} (Pose estimation error) Evaluate the attitude error
as $\tilde{R}_{o|k}=R_{y}\hat{R}_{k-1}^{\top}$ and the position error
as $\tilde{P}_{o|k}=P_{y|k}-\hat{P}_{k-1}$.\vspace{0.2cm}

\textcolor{blue}{\textbf{Step 4.}} (Thrust) Obtain $F_{k}$ and $\Im_{k}$ as in \eqref{eq:VTOL_F}
and \eqref{eq:VTOL_Thrust}
\begin{align*}
	F_{k} & =\ddot{P}_{d}-k_{\theta1}\psi(\theta_{k-1})-k_{\theta2}\psi(\dot{\theta}_{k-1})=[f_{1},f_{2},f_{3}]^{\top}\\
	\Im_{k} & =m||ge_{3}-F_{k}||
\end{align*}
with $\ddot{\theta}_{k}=-k_{\theta1}\psi(\theta_{k-1})-k_{\theta2}\psi(\dot{\theta}_{k-1})+k_{c3}(\hat{P}_{k-1}-P_{d|k}-\theta_{k-1})+k_{c4}(\hat{V}_{k-1}-V_{d|k}-\dot{\theta}_{k-1})$,
$\dot{\theta}_{k}=\dot{\theta}_{k-1}+\Delta t\ddot{\theta}_{k}$,
and $\theta_{k}=\theta_{k-1}+\Delta t\dot{\theta}_{k}$.\vspace{0.2cm}

\textcolor{blue}{\textbf{Step 5.}} (Prediction) $\hat{U}_{k}=\left[\begin{array}{ccc}
	[\hat{\Omega}_{k-1}\text{\ensuremath{]_{\times}}} & 0_{3\times1} & -\frac{\Im_{k}}{m}e_{3}\\
	0_{1\times3} & 0 & 0\\
	0_{1\times3} & 1 & 0
\end{array}\right]$ where $\hat{U}_{k}\in\mathcal{U}_{\mathcal{M}}$ and
\[
\hat{X}_{k|k-1}=\hat{X}_{k-1}\exp(\hat{U}_{k}\Delta t)
\]
$\exp(\cdot)$ denotes exponential of a matrix.\vspace{0.2cm}

\textcolor{blue}{\textbf{Step 6.}} (Correction factors) Evaluate the correction factors
as $w_{o}=-\gamma_{o}\tilde{R}_{o|k}^{\top}\mathbf{vex}(\boldsymbol{\mathcal{P}}_{a}(\tilde{R}_{o|k}))$, 

$w_{\Omega}=k_{o1}R_{y|k}^{\top}\mathbf{vex}(\boldsymbol{\mathcal{P}}_{a}(\tilde{R}_{o|k}))$,

$w_{V}=-\left[w_{\Omega}\right]_{\times}\hat{P}_{k-1}-k_{o2}\tilde{P}_{o|k}$,
and

$w_{a}=-\frac{\Im_{k}}{m}\hat{R}_{k-1}^{\top}(\mathbf{I}_{3}-\tilde{R}_{o|k}^{\top})e_{3}-ge_{3}-\left[w_{\Omega}\right]_{\times}\hat{V}_{k-1}-k_{o3}\tilde{P}_{o|k}$\vspace{0.2cm}

\textcolor{blue}{\textbf{Step 7.}} (Correction) $W=\left[\begin{array}{ccc}
	\left[w_{\Omega}\right]_{\times} & w_{V} & w_{a}\\
	0_{1\times3} & 0 & 0\\
	0_{1\times3} & 1 & 0
\end{array}\right]\in\mathcal{U}_{\mathcal{M}}$ and
\[
\hat{X}_{k}=\exp(-W\Delta t)\hat{X}_{k|k-1}
\]
where $\hat{P}_{k}=\hat{X}_{k}(1:3,4)$, $\hat{V}_{k}=\hat{X}_{k}(1:3,5)$,
and $\hat{R}_{k}=\hat{X}_{k}(1:3,1:3)^{\top}$.\vspace{0.2cm}

\textcolor{blue}{\textbf{Step 8.}} Follow the Appendix to evaluate the derivatives of
the intermediary control inputs $\dot{F}_{k}$ and $\ddot{F}_{k}$.
Also, $\Xi(F_{k})$ is evaluated as in \eqref{eq:VTOL_Q_Lambda} along
with its derivative $\dot{\Xi}(F)$.\vspace{0.2cm}

\textcolor{blue}{\textbf{Step 9.}} (Desired attitude) The desired unit-quaternion is
\[
q_{d0|k}=\sqrt{\frac{m}{2\Im_{k}}(g-f_{3})+\frac{1}{2}},\hspace{1em}q_{d|k}=\left[\begin{array}{c}
	\frac{m}{2\Im_{k}q_{d0}}f_{2}\\
	-\frac{m}{2\Im_{k}q_{d0}}f_{1}\\
	0
\end{array}\right]
\]
where $R_{d|k}=(q_{d0|k}^{2}-||q_{d|k}||^{2})\mathbf{I}_{3}+2q_{d|k}q_{d|k}^{\top}-2q_{d0|k}[q_{d|k}]_{\times}\in\mathbb{SO}\left(3\right)$.

\textcolor{blue}{\textbf{Step 10.}} (Desired angular velocity) Calculate $\Omega_{d|k}=\Xi(F_{k})\dot{F}_{k}$
and $\dot{\Omega}_{d|k}=\dot{\Xi}(F_{k})\dot{F}_{k}+\Xi(F_{k})\ddot{F}_{k}$
as in \eqref{eq:VTOL_Q_Omd} and \eqref{eq:VTOL_Q_Omd_dot}, respectively.\vspace{0.2cm}

\textcolor{blue}{\textbf{Step 11.}} (Rotational torque) Attitude error is evaluated
by $\tilde{R}_{c|k}=\hat{R}_{k}R_{d|k}^{\top}$ and the rotational
torque is calculated as
\begin{align*}
	\mathcal{T}_{k}= & k_{c1}\mathbf{vex}(\boldsymbol{\mathcal{P}}_{a}(\tilde{R}_{c|k}))-k_{c2}(\tilde{R}_{o|k}\hat{\Omega}_{k}-\tilde{R}_{c|k}\Omega_{d|k})\\
	& +\left[\tilde{R}_{c|k}\Omega_{d|k}\right]_{\times}J\tilde{R}_{c|k}\Omega_{d|k}+J\tilde{R}_{c|k}\dot{\Omega}_{d|k}
\end{align*}

\textcolor{blue}{\textbf{Step 12.}} (Angular velocity estimate) The angular velocity
estimate is evaluated by
\[
\hat{\Omega}_{k+1}=\hat{\Omega}_{k}+\Delta t\hat{J}^{-1}([\hat{J}\hat{\Omega}_{k}]_{\times}\hat{\Omega}_{k}+\hat{\mathcal{T}}_{k}-\hat{J}[\hat{\Omega}_{k}]_{\times}\hat{R}_{k}w_{\Omega}+w_{o})
\]
where $\hat{\mathcal{T}}_{k}=\tilde{R}_{o|k}^{\top}\mathcal{T}_{k}$
and $\hat{J}=\tilde{R}_{o|k}^{\top}J\tilde{R}_{o|k}$.\vspace{0.2cm}

\textcolor{blue}{\textbf{Step 13.}} Set $k=k+1$, and go to \textcolor{blue}{\textbf{Step 2}}.

\section{Simulation and Experimental Results \label{sec:SE3_Simulations}}

\subsection{Simulation Results}

This subsection presents the output performance of the proposed observer-based
controller for a 6 DoF VTOL-UAV. The observing and tracking control
capabilities are tested in a discrete form at a low sampling rate
of 1000 Hz against unknown random noise and constant bias corrupting
the measurements. Consider the mass and the inertia matrix of the
VTOL-UAV to be $m=2.5\,\text{kg}$ and $J={\rm diag}(0.14,0.2,0.12)\,\text{kg}.\text{m}^{2}$,
respectively. Let the desired trajectory be 
\[
P_{d}=6\left[\begin{array}{c}
	\sin(0.2t)\\
	\sin(0.2t)\cos(0.2t)\\
	\frac{1}{6}(4+0.15t)
\end{array}\right]\,\text{m}
\]
The total time is set to 50 seconds. Let the true initial orientation,
angular velocity, position, and linear velocity of the VTOL-UAV be
\begin{align*}
	R_{0} & =\left[\begin{array}{ccc}
		-0.2712 & -0.7130 & 0.6466\\
		0.8655 & -0.4746 & -0.1603\\
		0.4212 & 0.5162 & 0.7458
	\end{array}\right]\\
	\Omega_{0} & =[0,0,0]^{\top}\,\text{rad/sec}\\
	P_{0} & =[-2,-1,0]^{\top}\,\text{m}\\
	V_{0} & =[0,0,0]^{\top}\,\text{m/\ensuremath{\sec}}
\end{align*}
Let the estimated initial orientation, angular velocity, position,
and linear velocity of the vehicle be
\begin{align*}
	\hat{R}_{0} & =\mathbf{I}_{3}\\
	\hat{\Omega}_{0} & =\hat{P}_{0}=\hat{V}_{0}=[0,0,0]^{\top}
\end{align*}
Let $\hat{\theta}_{0}=\dot{\hat{\theta}}_{0}=[0,0,0]^{\top}$. From
\eqref{eq:VTOL_VR}, $y_{i}^{\mathcal{B}}=R{\rm v}_{i}^{\mathcal{I}}+b_{i}^{\mathcal{B}}+n_{i}^{\mathcal{B}}$
for $i=1,2$, let 
\[
\begin{cases}
	{\rm v}_{1}^{\mathcal{I}} & =[0,0,1]^{\top}\\
	{\rm v}_{2}^{\mathcal{I}} & =[1,-1,-1]^{\top}
\end{cases}
\]
\begin{figure*}
	\centering{}\includegraphics[scale=0.35]{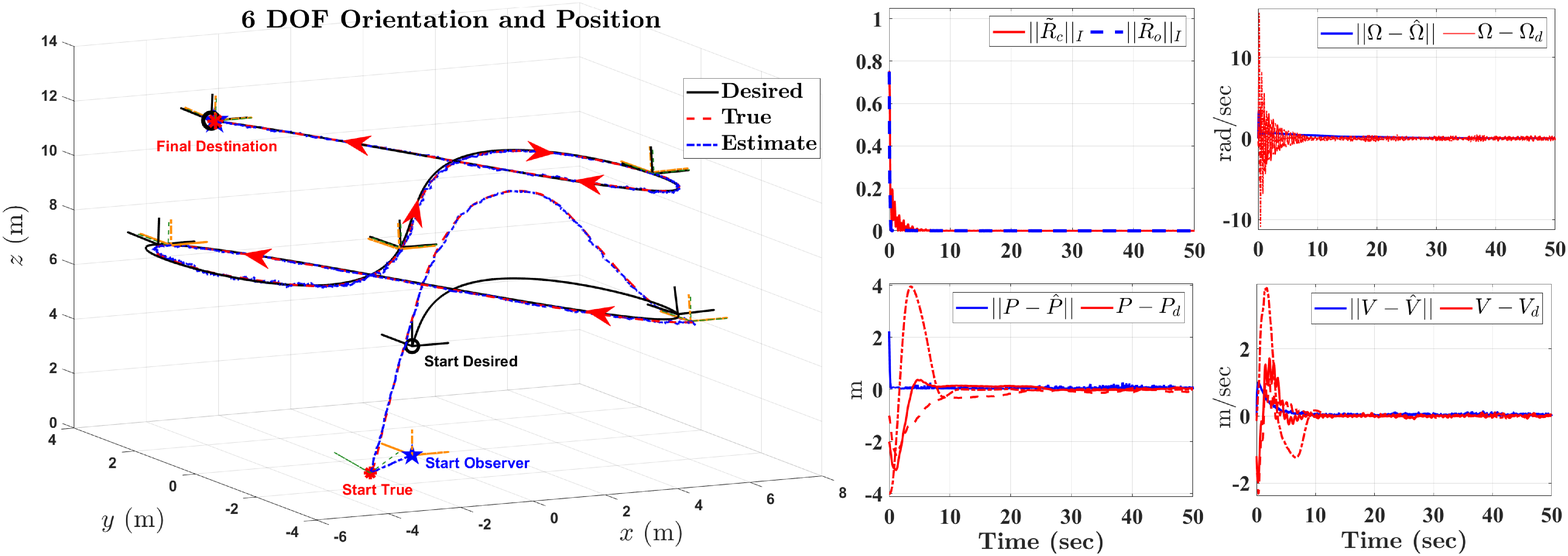}\caption{Output performance of the observer-based controller implemented using
		a 6 DoF VTOL-UAV. The left portion illustrates the desired position
		(black solid line), the true position (red dashed line), and the estimated
		position (blue center line). The VTOL-UAV orientation (roll, yaw,
		and pitch) is shown as a black solid line, a green dashed line, and
		an orange center line corresponding to the desired, true, and estimated
		orientation, respectively. Convergence of the error trajectories is
		demonstrated in the right portion where blue and red denote the error
		components between the true and the observed data, and between the
		true and the desired data, respectively.}
	\label{fig:ObsvCont}
\end{figure*}
with $b_{1}^{\mathcal{B}}=[0,0,-0.15]^{\top}$, $b_{2}^{\mathcal{B}}=[0.1,0.09,-0.11]^{\top}$,
and normally distributed noise $n_{1}^{\mathcal{B}}$ and $n_{2}^{\mathcal{B}}$
having a zero mean and a STD of 0.05, in other words $\mathcal{N}(0,0.05)$. The
attitude is reconstructed using SVD \cite{markley1988attitude} where
${\rm v}_{3}^{\mathcal{I}}={\rm v}_{1}^{\mathcal{I}}\times{\rm v}_{2}^{\mathcal{I}}$
and $y_{3}^{\mathcal{B}}=y_{1}^{\mathcal{B}}\times y_{2}^{\mathcal{B}}$:
\begin{equation}
	\begin{cases}
		{\bf r}_{i} & =\frac{{\rm v}_{i}^{\mathcal{I}}}{||{\rm v}_{i}^{\mathcal{I}}||},\hspace{1em}{\bf y}_{i}=\frac{y_{i}^{\mathcal{B}}}{||y_{i}^{\mathcal{B}}||},\hspace{1em}i=1,2,\ldots,N\\
		B & =\sum_{i=1}^{n}s_{i}{\bf y}_{i}{\bf r}_{i}^{\top}=USV^{\top}\\
		U_{+} & =U\cdot diag(1,1,\det(U))\\
		V_{+} & =V\cdot diag(1,1,\det(V))\\
		R_{y} & =V_{+}U_{+}^{\top}
	\end{cases}\label{eq:SVD}
\end{equation}
Consider a group of seven non-collinear randomly distributed landmarks ($N_{2}=7$)
satisfying Assumption \ref{Assum:VTOL_1Landmark} item A1. Let $s_{j}=1\forall j=1,2,\ldots,7$ and $s_{c}=\sum_{j=1}^{N_{2}}s_{j}$. The landmark
measurements are defined as in \eqref{eq:VTOL_VRP} and incorporate
added constant bias and normally distributed noise ($\mathcal{N}(0,0.05)$).
The position at each time instant is reconstructed as (see \eqref{eq:VTOL_R_Weighted}):
\begin{equation}
	\begin{cases}
		p_{c} & =\frac{1}{s_{c}}\sum_{j=1}^{N_{2}}s_{j}p_{j}^{\mathcal{I}},\hspace{1em}z_{c}=\frac{1}{s_{c}}\sum_{j=1}^{N_{2}}s_{j}z_{j}^{\mathcal{B}}\\
		P_{y} & =p_{c}-R_{y}^{\top}z_{c}
	\end{cases}\label{eq:Rec}
\end{equation}
 For more details of attitude and pose reconstruction visit \cite{hashim2019SO3Wiley, hashim2020SE3Stochastic}. Let the design parameters be selected as follows: $\gamma_{o}=0.1$,
$k_{o1}=10$, $k_{o2}=10$, $k_{o3}=5$, $k_{c1}=10$, $k_{c2}=0.1$,
$k_{c1}=2$, $k_{c2}=4$, $k_{\theta1}=1$, and $k_{\theta2}=1$.

Fig. \ref{fig:ObsvCont} demonstrates the output performance of
a VTOL-UAV guided by the observer-based controller. Fig. \ref{fig:ObsvCont}
shows robust, strong, fast, and smooth tracking performance of the
proposed observer-based controller and its ability to guide the VTOL-UAV
to the desired final destination, despite a large initial error. The
error is shown to rapidly converge from large initial values to the
neighborhood of the attractive equilibrium point. Fig. \ref{fig:ControlInput}
depicts the bounded rotational torque input and the thrust input.

\begin{figure}
	\centering{}\includegraphics[scale=0.32]{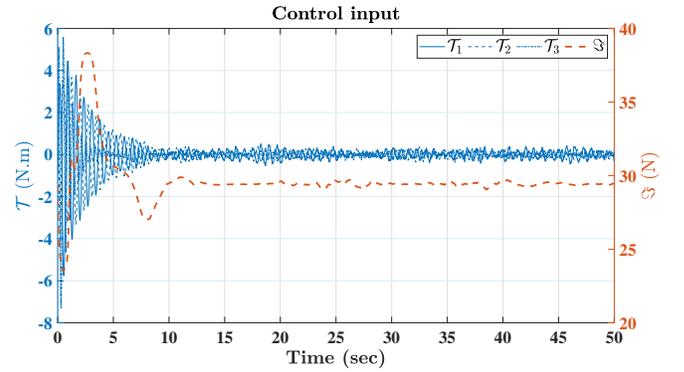}\caption{VTOL-UAV control input.}
	\label{fig:ControlInput}
\end{figure}

\subsection{Experimental Results}

To further validate the observing capabilities, the proposed observer with the appropriate modification 
has been tested using the EuRoC real-world dataset \cite{burri2016euroc}
that includes the ground truth of a real-life quadrotor flight trajectory,
stereo images, and IMU data. The ADIS16448 IMU collected data at a
sampling rate of 200 Hz. The MT9V034 sensor collected stereo images
at a sampling rate of 20 Hz which were subsequently undistorted with
the camera parameters and calibrated using a Stereo Camera Calibrator
in MATLAB. For more details about the EuRoC dataset visit \cite{burri2016euroc}.
The landmarks were tracked using minimum eigenvalue landmark detection
through Kanade-Lucas-Tomasi (KLT) feature tracker \cite{shi1994good},
see Fig. \ref{fig:Features}. Due to the fact that the dataset had
no landmark information, a set of landmarks was generated from the
stereo images using $p_{j}^{\mathcal{I}}=R^{\top}z_{j}^{\mathcal{B}}+P$
where $P$ and $R$ denote ground truth position and orientation,
respectively%
\begin{comment}
	, and $b_{j}^{\mathcal{B}}$ and $n_{j}^{\mathcal{B}}=\mathcal{N}(0,0.1)$
	denote constant bias and random noise, respectively, defined in \eqref{eq:VTOL_VRP}
\end{comment}
. For the purposes of the experiment, the maximum number of detected
landmarks was limited to 50. The coordinates of the landmark camera
frame (cam0 EuRoC dataset) were transformed to the vehicle frame using
the calibration matrix included in the dataset.

\begin{figure}
	\centering{} \includegraphics[scale=0.2]{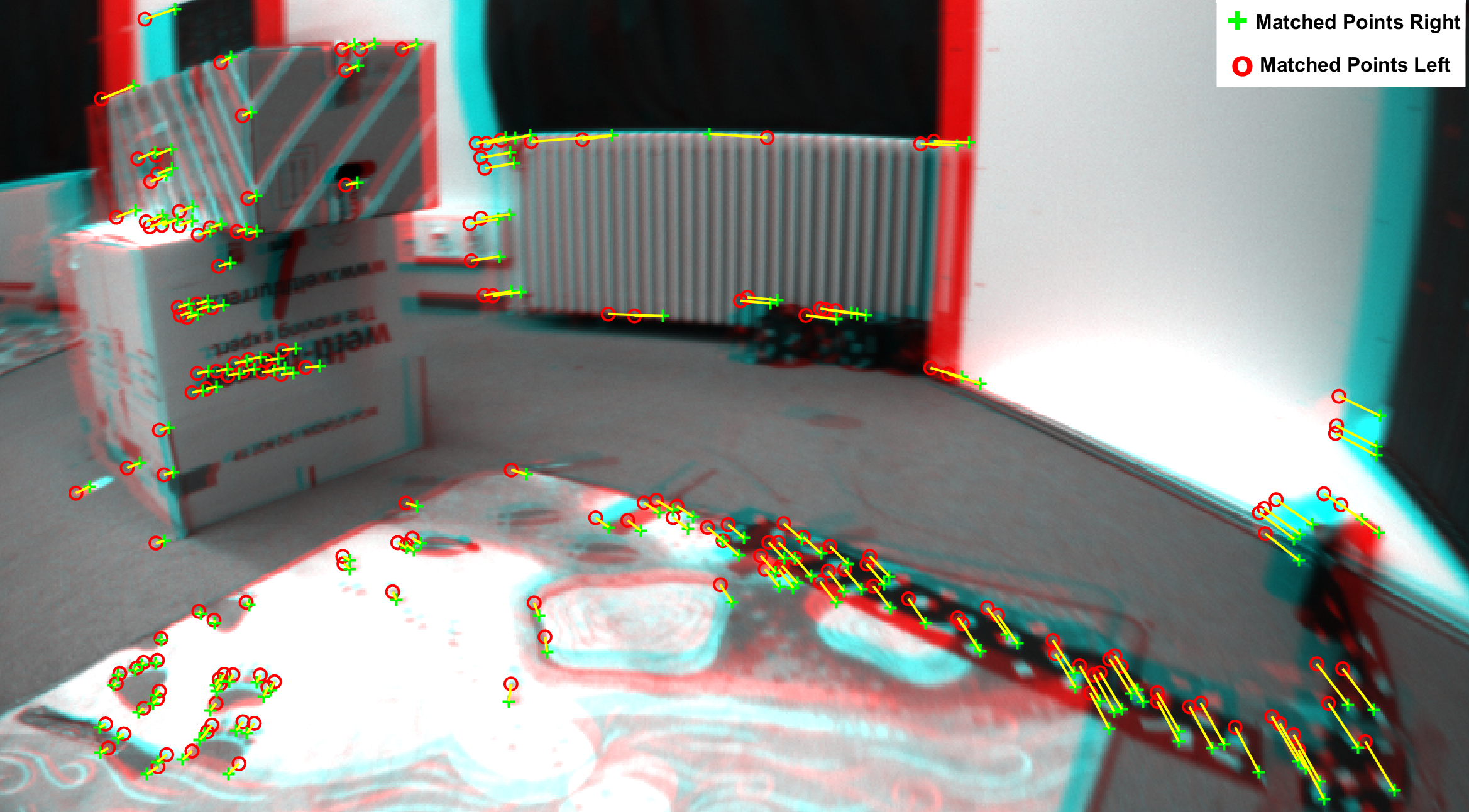}\caption{An example of landmark detection and tracking. Photographs are obtained
		from the EuRoC dataset \cite{burri2016euroc}.}
	\label{fig:Features}
\end{figure}

The experimental results shown in Fig. \ref{fig:Obsv} reveal strong
tacking capability of the proposed observer demonstrating the robustness
of the proposed approach. To summarize, Fig. \ref{fig:ObsvCont}
and \ref{fig:Obsv} illustrate the effectiveness of the proposed observer-based
controller to observe the unknown motion parameters, namely orientation,
angular velocity, position, and linear velocity, while tracking the
vehicle along the desired trajectory.

\begin{figure}
	\centering{}\includegraphics[scale=0.29]{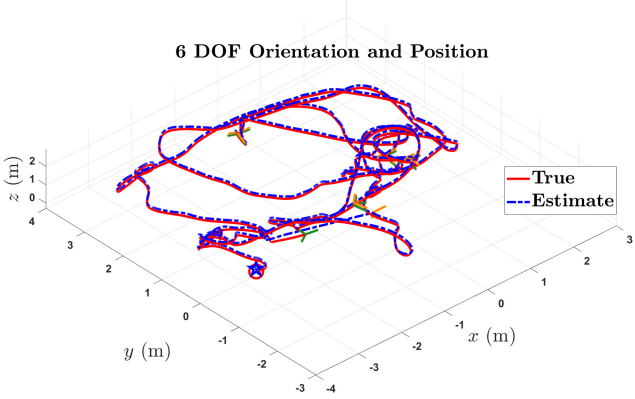}\caption{Experimental validation using the Vicon Room 2 02 dataset. Output
		performance of the proposed observer for a 6 DoF UAV.}
	\label{fig:Obsv}
\end{figure}

\section{Conclusion \label{sec:SE3_Conclusion}}

This paper addressed the estimation and control of the motion parameters,
namely attitude, angular velocity, position, and linear velocity,
in application to a six degrees of freedom (6 DoF) Vertical Take-Off
and Landing Unmanned Aerial Vehicle (VTOL-UAV). The newly proposed
observer-based controller as well as the observer that lies at its
foundation are both characterized by almost global exponential stability
of the closed loop error signals regardless of the initial condition.
The proposed approach does not require a gyroscope and Global Positioning
Systems signals. Use of measurements obtained by a low-cost measurement
unit at a low sampling rate do not compromise the performance of the
proposed observer-based controller. On the contrary, as has been revealed
by simulation and experimental results, the proposed approach is distinguished
by accurate observation and robust tracking control to the desired
trajectory of the VTOL-UAV motion parameters, namely attitude, angular
velocity, position, and linear velocity.

\section*{Acknowledgment}
	The authors would like to thank \textbf{Maria Shaposhnikova} for proofreading
	the article.

\subsection*{Appendix A\label{subsec:Appendix-A}}

Let $\theta\in\mathbb{R}^{3}$ be an auxiliary variable defined in
\eqref{eq:VTOL_Q_E}. Define the mapping of 
\[
\psi(\theta_{i})=\frac{\exp(\theta_{i})-\exp(-\theta_{i})}{\exp(\theta_{i})+\exp(-\theta_{i})},\hspace{1em}\forall i=1,2,3
\]
and 
\[
\psi(\dot{\theta}_{i})=\frac{\exp(\dot{\theta}_{i})-\exp(-\dot{\theta}_{i})}{\exp(\dot{\theta}_{i})+\exp(-\dot{\theta}_{i})},\hspace{1em}\forall i=1,2,3
\]
such that $\psi(\theta)=[\psi(\theta_{1}),\psi(\theta_{2}),\psi(\theta_{3})]^{\top}\in\mathbb{R}^{3}$
and $\psi(\dot{\theta})=[\psi(\dot{\theta}_{1}),\psi(\dot{\theta}_{2}),\psi(\dot{\theta}_{3})]^{\top}\in\mathbb{R}^{3}$.
One can easily show that $\frac{d}{dt}\psi(\theta_{i})=(1-\psi(\theta_{i})^{2})\dot{\theta}_{i}$,
$\frac{d^{2}}{dt^{2}}\psi(\theta_{i})=(1-\psi(\theta_{i})^{2})(\ddot{\theta}_{i}-2\psi(\theta_{i})\dot{\theta}_{i}^{2})$,
$\frac{d}{dt}\psi(\dot{\theta}_{i})=(1-\psi(\dot{\theta}_{i})^{2})\ddot{\theta}_{i}$,
and $\frac{d^{2}}{dt^{2}}\psi(\dot{\theta}_{i})=(1-\psi(\dot{\theta}_{i})^{2})(\theta_{i}^{(3)}-2\psi(\dot{\theta}_{i})\ddot{\theta}_{i}^{2})$.
Therefore, the first and the second derivatives of $F$ are as follows:
\begin{align*}
	\dot{F} & =P_{d}^{(3)}-k_{\theta1}\frac{d}{dt}\psi(\theta)-k_{\theta2}\frac{d}{dt}\psi(\dot{\theta})\\
	\ddot{F} & =P_{d}^{(4)}-k_{\theta1}\frac{d^{2}}{dt^{2}}\psi(\theta)-k_{\theta2}\frac{d^{2}}{dt^{2}}\psi(\dot{\theta})
\end{align*}
with
\begin{align*}
	\theta^{(3)}= & -k_{\theta1}\frac{d}{dt}\psi(\theta)-k_{\theta2}\frac{d}{dt}\psi(\dot{\theta})+k_{c1}(\dot{\hat{P}}-\dot{P}_{d}-\dot{\theta})\\
	& +k_{c2}(\dot{\hat{V}}-\dot{V}_{d}-\ddot{\theta})
\end{align*}
where $P_{d}^{(3)}=\frac{d^{3}}{dt^{3}}P_{d}$, $P_{d}^{(4)}=\frac{d^{4}}{dt^{4}}P_{d}$,
and $\theta^{(3)}=\frac{d^{3}}{dt^{3}}\theta$. In addition, $\dot{\alpha}_{1}=\frac{1}{\alpha_{1}}[f_{1},f_{2},(f_{3}-g)]^{\top}\dot{F}$
and $\dot{\alpha}_{2}=\dot{\alpha}_{1}-\dot{f}_{3}$ with $\dot{F}=[\dot{f}_{1},\dot{f}_{2},\dot{f}_{3}]^{\top}$.

\subsection*{Appendix B\label{subsec:Appendix-B}}
\begin{center}
	\textbf{Quaternion Representation of the Observer-based Controller}
	\par\end{center}

\noindent Recall Section \ref{sec:Preliminaries-and-Math} and let
$Q_{y}=[q_{y0},q_{y}^{\top}]^{\top}\in\mathbb{S}^{3}$ be the reconstructed
attitude, obtained for instance, using QUEST algorithm \cite{shuster1981three}.
Define the reconstructed attitude $\mathcal{R}_{y}:\mathbb{S}^{3}\rightarrow\mathbb{SO}\left(3\right)$
as (see \eqref{eq:NAV_Append_SO3}):
\[
\mathcal{R}_{y}=(q_{y0}^{2}-||q_{y}||^{2})\mathbf{I}_{3}+2q_{y}q_{y}^{\top}-2q_{y0}[q_{y}]_{\times}\in\mathbb{SO}\left(3\right)
\]
Define $\hat{Q}=[\hat{q}_{0},\hat{q}^{\top}]^{\top}\in\mathbb{S}^{3}$
as the estimate of $Q=[q_{0},q^{\top}]^{\top}\in\mathbb{S}^{3}$.
Define the estimated attitude $\hat{\mathcal{R}}:\mathbb{S}^{3}\rightarrow\mathbb{SO}\left(3\right)$
as
\[
\hat{\mathcal{R}}=(\hat{q}_{0}^{2}-||\hat{q}||^{2})\mathbf{I}_{3}+2\hat{q}\hat{q}^{\top}-2\hat{q}_{0}[\hat{q}]_{\times}\in\mathbb{SO}\left(3\right)
\]
Let the error in estimation be $\tilde{Q}_{o}=\hat{Q}^{-1}\odot Q_{y}=[\tilde{q}_{o0},\tilde{q}_{o}^{\top}]^{\top}\in\mathbb{S}^{3}$.
Define the error between the estimated and the true attitude $\tilde{\mathcal{R}}_{o}:\mathbb{S}^{3}\rightarrow\mathbb{SO}\left(3\right)$
as
\[
\tilde{\mathcal{R}}_{o}=(\tilde{q}_{o0}^{2}-||\tilde{q}_{o}||^{2})\mathbf{I}_{3}+2\tilde{q}_{o}\tilde{q}_{o}^{\top}-2\tilde{q}_{o0}[\tilde{q}_{o}]_{\times}\in\mathbb{SO}\left(3\right)
\]
Let the error in control be $\tilde{Q}_{c}=Q_{d}^{-1}\odot\hat{Q}=[\tilde{q}_{c0},\tilde{q}_{c}^{\top}]^{\top}\in\mathbb{S}^{3}$
and define the error between the desired and the true attitude $\tilde{\mathcal{R}}_{c}:\mathbb{S}^{3}\rightarrow\mathbb{SO}\left(3\right)$
as
\[
\tilde{\mathcal{R}}_{c}=(\tilde{q}_{c0}^{2}-||\tilde{q}_{c}||^{2})\mathbf{I}_{3}+2\tilde{q}_{c}\tilde{q}_{c}^{\top}-2\tilde{q}_{c0}[\tilde{q}_{c}]_{\times}\in\mathbb{SO}\left(3\right)
\]
The quaternion representation of the observer in \eqref{eq:VTOL_ObsvCompact}-\eqref{eq:VTOL_Vestdot}
is as below:
\[
\begin{cases}
	\Phi & =\left[\begin{array}{cc}
		0 & -\hat{\Omega}^{\top}\\
		\hat{\Omega} & -[\hat{\Omega}]_{\times}
	\end{array}\right],\hspace{1em}\Psi=\left[\begin{array}{cc}
		0 & -w_{\Omega}^{\top}\\
		w_{\Omega} & [w_{\Omega}]_{\times}
	\end{array}\right]\\
	\dot{\hat{Q}} & =\frac{1}{2}(\Phi-\Psi)\hat{Q}\\
	\hat{J}\dot{\hat{\Omega}} & =[\hat{J}\hat{\Omega}]_{\times}\hat{\Omega}+\hat{\mathcal{T}}-\hat{J}[\hat{\Omega}]_{\times}\hat{\mathcal{R}}w_{\Omega}+w_{o}\\
	\dot{\hat{P}} & =\hat{V}-[w_{\Omega}]_{\times}\hat{P}-w_{V}\\
	\dot{\hat{V}} & =-\frac{\Im}{m}\hat{\mathcal{R}}^{\top}e_{3}-[w_{\Omega}]_{\times}\hat{V}-w_{a}
\end{cases}
\]
where $\hat{\mathcal{T}}=\tilde{\mathcal{R}}_{o}^{\top}\mathcal{T}$,
$\hat{J}=\tilde{\mathcal{R}}_{o}^{\top}J\tilde{\mathcal{R}}_{o}$,
and 
\[
\begin{cases}
	w_{o} & =-2\gamma_{o}\tilde{q}_{o0}\tilde{\mathcal{R}}_{o}^{\top}\tilde{q}_{o}\\
	w_{\Omega} & =2k_{o1}\tilde{q}_{o0}\mathcal{R}_{y}^{\top}\tilde{q}_{o}\\
	w_{V} & =-\left[w_{\Omega}\right]_{\times}\hat{P}-k_{o2}\tilde{P}_{o},\hspace{1em}\tilde{P}_{o}=P_{y}-\hat{P}\\
	w_{a} & =-\frac{\Im}{m}\hat{\mathcal{R}}^{\top}(\mathbf{I}_{3}-\tilde{\mathcal{R}}_{o}^{\top})e_{3}-ge_{3}-\left[w_{\Omega}\right]_{\times}\hat{V}-k_{o3}\tilde{P}_{o}
\end{cases}
\]
where $\boldsymbol{\Upsilon}(\tilde{\mathcal{R}}_{o})=2\tilde{q}_{o0}\tilde{q}_{o}$
(see \cite{hashim2019AtiitudeSurvey,shuster1993survey}). The quaternion
representation of the control laws in \eqref{eq:VTOL_Tau}-\eqref{eq:VTOL_Thrust}
is as below:
\[
\begin{cases}
	\mathcal{T} & =2k_{c1}\tilde{q}_{c0}\tilde{q}_{c}-k_{c2}(\tilde{\mathcal{R}}_{o}\hat{\Omega}-\tilde{\mathcal{R}}_{c}\Omega_{d})+J\tilde{\mathcal{R}}_{c}\dot{\Omega}_{d}\\
	& \hspace{1em}+[\tilde{\mathcal{R}}_{c}\Omega_{d}]_{\times}J\tilde{\mathcal{R}}_{c}\Omega_{d}\\
	\ddot{\theta} & =-k_{\theta1}\psi(\theta)-k_{\theta2}\psi(\dot{\theta})+k_{c3}(\hat{P}-P_{d}-\theta)\\
	& \hspace{1em}+k_{c4}(\hat{V}-V_{d}-\dot{\theta})\\
	F & =\ddot{P}_{d}-k_{\theta1}\psi(\theta)-k_{\theta2}\psi(\dot{\theta})\\
	\Im & =m||ge_{3}-F||
\end{cases}
\]
where $\boldsymbol{\Upsilon}(\tilde{\mathcal{R}}_{c})=2\tilde{q}_{c0}\tilde{q}_{c}$.
To know more about attitude paramterization and mapping visit \cite{hashim2019AtiitudeSurvey}.

\bibliographystyle{IEEEtran}
\bibliography{bib_VTOL_ObsvCont}
% name your BibTeX data base

%
%\vspace{170pt}
%
%
%
%
%\section*{AUTHOR INFORMATION}
%\vspace{10pt}
%
%	{\bf Hashim A. Hashim} (Senior Member, IEEE) is currently an Assistant Professor at the Department of Mechanical and Aerospace Engineering at Carleton University, Ottawa, Ontario, Canada. He received the B.Sc. degree in Mechatronics, Department of Mechanical Engineering from Helwan University, Cairo, Egypt, the M.Sc. in Systems and Control Engineering, Department of Systems Engineering from King Fahd University of Petroleum \& Minerals, Dhahran, Saudi Arabia, and the Ph.D. in Robotics and Control, Department of Electrical and Computer Engineering at Western University, Ontario, Canada.\\
%	His current research interests include Guidance, navigation and control, simultaneous localization and mapping, vision-aided inertial navigation systems, control of multi-agent systems, and optimization techniques.

\end{document}